%% file: main-ver2.tex
\documentclass[journal,twoside,web]{ieeecolor}

\usepackage{generic}
\usepackage{cite}

\usepackage{graphics}
\graphicspath{{fig/}}

\usepackage{amssymb}
\usepackage{mathtools}
\usepackage{amsfonts}
\usepackage{mathrsfs} 
\usepackage{mathabx}
\usepackage{bm}

\let\labelindent\relax
\usepackage[shortlabels]{enumitem}

\usepackage{tabularray}
\usepackage[]{diagrams}
\usepackage{pgfmath, pgfplots}
\usepackage{xcolor}
\usepackage{tikz}
\usetikzlibrary{arrows.meta, positioning, calc, matrix}
\usepackage{subcaption}
\usetikzlibrary{patterns}

\def\BibTeX{{\rm B\kern-.05em{\sc i\kern-.025em b}\kern-.08em
    T\kern-.1667em\lower.7ex\hbox{E}\kern-.125emX}}

\input{cmds.tex}

\begin{document}

\title{
On the Output Redundancy of LTI Systems:\\ A Geometric Approach with Application to Privacy
}

\author{
    Guitao Yang, 
    Alexander J. Gallo, 
    Angelo Barboni,
    Riccardo M.G. Ferrari, \\
    Andrea Serrani, \IEEEmembership{Member, IEEE}, 
    and Thomas Parisini, \IEEEmembership{Fellow, IEEE}
    \thanks{This work has been partially supported by European Union's Horizon 2020 research and innovation program under grant agreement no. 739551 (KIOS CoE).}
    \thanks{G. Yang is with the Department of Electrical and Electronic Engineering, Imperial College London, London SW7 2AZ, UK (e-mail: guitao.yang@imperial.ac.uk).}
    \thanks{A. J. Gallo and R. M. G. Ferrari are with the Delft Centre for Systems and Control, Technical University of Delft, Delft, The Netherlands (e-mail: a.j.gallo@tudelft.nl, r.ferrari@tudelft.nl)}
    \thanks{A. Barboni is with the Zurich Insurance Group, Zurich, Switzerland (e-mail: angelobarboni.eng@gmail.com)}
    \thanks{A. Serrani is with the Department of Electrical and Computer Engineering, The Ohio State University, Columbus, OH 43210, USA  (e-mail: serrani.1@osu.edu).}
    \thanks{T. Parisini is with the Department of Electrical and Electronic Engineering, Imperial College London, SW72AZ London, U.K., He is also with the Department of Electronic Systems, Aalborg University, Denmark, and with the Department of Engineering and Architecture, University of
Trieste, Italy (e-mail: t.parisini@imperial.ac.uk).}
}

\maketitle
\begin{abstract}
    This paper examines the properties of output-redundant systems, that is, systems possessing a larger number of outputs than inputs, through the lenses of the  geometric approach of Wonham et al. We begin by formulating a simple output allocation synthesis problem, which involves ``concealing" input  information from a malicious eavesdropper having access to the system output, while still allowing for a legitimate user to reconstruct it. It is shown that the solvability of this problem requires the availability of a redundant set of outputs. This very problem is instrumental to unveiling the fundamental geometric properties of output-redundant systems, which form the basis for our subsequent constructions and results. As a direct application, we demonstrate how output allocation can be employed to effectively protect the information of input information from certain output eavesdroppers with guaranteed results.
\end{abstract}
\begin{IEEEkeywords}
    Output Redundancy, System Invertibility, Geometric Approach, LTI Systems, 
\end{IEEEkeywords}
\section{Introduction}
In this paper, the properties and the structure of output-redundant systems are analyzed in a geometric context, through the examination of the problem of protecting  sensitive exogenous input information from a malicious eavesdropper having access to measured outputs. 

The progressive increase in complexity of Cyber-Physical Systems (CPS), stemming from the integration of widespread communication and computation technologies, have made them prime targets for all sorts of cyber-attacks and adversarial interventions~\cite{osti_1505628}. Among the possible actions taken by malicious agents are eavesdropping attacks~\cite{chong2019tutorial}, i.e., attacks which breach the confidentiality of the information being transmitted within the control system, enabling them to access private user information.
This problem is particularly relevant for cases, such as smart grids or heating, ventilation, and air conditioning (HVAC) systems, where measurement outputs may contain sensitive information about user behavior, such as energy consumption, or room occupancy \cite{kumar2019smart,chen2018building,alisic2020ensuring}.
Among the schemes proposed by the control-theoretic community to achieve privacy  is that of achieving obfuscation of the relevant information by affine transformations of the output, whereby a disturbance is injected onto the transmitted signal. Without the pretense of providing a complete overview, examples of this strategy are reported in \cite{sultangazin2020symmetries,naseri2022privacy,hayati2022privacy}. 

In this paper, the very problem of obfuscating exogenous inputs is used to introduce and investigate fundamental geometric concepts associated to the property of output redundancy in linear systems.
To set the stage for our discussion, we consider LTI systems of the form
\begin{equation}
    \label{eq:system}
    \begin{aligned}
        x^+ &= A x  + B u \\
        y &= C x
    \end{aligned}
\end{equation}
where $x, u$ and $y$ are the state, input and output, respectively, seen as elements of appropriate finite-dimensional vector spaces in the field of real numbers. 
The operator $+$ in $x^+$ denotes the evolution of the state variable in a general sense, which encompasses continuous-time evolution ($dx/dt$), discrete-time evolution $x(t+1)$ or even variables evolving at non-equally spaced time instances ($x(t_{k+1})$). In~\eqref{eq:system}, $A$, $B$ and $C$ represent maps among vector spaces. Assuming access to the measured output $y$, as well as the system parameters $A$, $B$, and $C$, under certain system invertibility conditions an attacker may recover the exact system input, which constitutes the private and sensitive information being transmitted within the control systems. However, if the output signal is modified as 
\begin{equation*}
    {y}_d = y + Dd,
\end{equation*}
where $d$ represents a disturbance intentionally injected through the map $D$, the eavesdropper might be misled by the resulting corruption of the available information, and may lose its capability to recover sensitive information. 
In the considered scenario, we want to evaluate this simple strategy provided that a legitimate user is able to reconstruct the required information from the corrupted data, whereas the eavesdropper can not. 
The considered setup is illustrated in the schematic diagram in Figure~\ref{fig:sche}, where an ``inverse system'' (for lack of better terminology) is introduced to reconstruct the exogenous input. 
%

For this to be possible, we show that it is necessary for the system to be {\em output-redundant}, i.e., that the number of scalar measurements, $p\in\mathbb{N}$, be strictly larger than the dimension $m\in\mathbb{N}$ of the input signal $u(\cdot)$ to be concealed from the eavesdropper.
We then investigate the geometric properties of the ensuing output-redundant system~\eqref{eq:system}, evaluating under what conditions it is possible to design the map $D$, such that two objectives can be achieved simultaneously: firstly, the input and state can be reconstructed by some legitimate user (a \textit{defender}) with some knowledge of $D$; secondly, that a malicious agent, 
without knowledge of $D$, obtains a biased estimate of the input.
The problem is posed as an {\em output allocation problem} via the selection of the {\em output allocation map $D$}, conducted on the basis of the geometric interpretation of output redundancy. 
Intuitively, the difference between the number of inputs and outputs of a system provides more ``flexibility'' to inject a disturbance capable of obfuscating the input from a malicious agent, without affecting reconstruction by a legitimate defender. 
We investigate the properties associated with output redundancy of a system in order to make this notion of ``flexibility'' precise.
We note that the need for ``flexibility'' finds correspondence in the method proposed in \cite{hayati2022privacy}, whereby it is the lifting of the output signal to be made private into a larger space (and therefore the inclusion of additional ``flexibility'') that enables the design of a scheme that ensures privacy, without any performance loss.
%

To make the presentation reasonably self-contained, apart from output redundancy, additional properties will be assumed to hold for system~\eqref{eq:system}, namely minimality and left-invertibility.

\subsection{Brief Literature Review}
Existing studies have examined output redundancy from various perspectives. A somewhat customary assumption dictates that the sensors measuring the system state are duplicated, with subsets of these sensors potentially susceptible to faults \cite{yang2019sensor,winkler2021using,yang2022sensor}. In the cited references, output redundancy is leveraged to mitigate the impact of faulty sensors, whereas the overall system performance continues to rely on the non-faulty sensors. The studies \cite{ren2023effects,cristofaro2024adaptive} consider setups where certain measurements are linearly dependent on others, which presents a broader scenario compared to duplicated sensors. Specifically, \cite{ren2023effects} demonstrates that the inclusion of redundant sensors improves estimation performance by reducing the variance of the estimated state error. On the other end, in \cite{cristofaro2024adaptive}, the authors propose a methodology to dynamically select a weighted combination of available measurements to reject constant biases.
In contrast to the notion of output redundancy explored in the aforementioned works, \cite{corona2016some} defines output redundancy as a dual concept of input redundancy, carrying over a taxonomy that distinguishes between the notions of strong and weak output redundancy. In that paper, a system is termed to be strongly output-redundant when the output matrix lacks full rank, whereas weak output redundancy occurs when the number of outputs exceeds the number of inputs. While this perspective enriches the concept of output redundancy beyond the scope of previous works, which primarily fall under the strong output redundancy paradigm, the focus of the authors in \cite{corona2016some} is centered on the output regulation problem (hence, in a certain sense on right-invertibility), and does not consider aspects related to left-invertibility, which we believe to be of fundamental importance in the considered setup.
\begin{table*}[htbp]
    \centering
    \SetTblrInner{rowsep=1.2mm}
    \begin{tblr}{c|c|c}
        & Input allocation & Output allocation \\
        \hline
        System & 
            $ \left\lbrace
               \begin{aligned} 
                    x^+ &= Ax + Bu \\
                    y &= Cx    
             \end{aligned}\right. $ & 
             $ \left\lbrace
               \begin{aligned} 
                    \hat x^+ &= A \hat x  - L ( y_d - \hat y) \\
                    \hat y &= C \hat x      
             \end{aligned}\right. $
             \\
        Control variable & input $u$ & output error $y_d - \hat y$ \\
        Control objective & $\| y_{\text{ref}} - y \| \rightarrow$ small & $ \| u - \hat u \|  \rightarrow$ small\\
        Allocation mechanism & $u = u_c + B^\perp d$ \cite{zaccarian2009dynamic} & $y_d = y + Dd$ (this paper) \\
        \hline
    \end{tblr}
    \caption{Duality between dynamic input and output allocation. In the input allocation section, $u_{c}$ denotes the input provided by a given tracking controller, and $y_{\rm ref}$ a reference output to be tracked.}
    \label{fig:table_duality}
\end{table*}
%
To address this gap, this paper provides a study of output redundancy and ensuing properties from a coordinate-free geometric perspective; algorithms generated from this approach apply to a variety of formulations, encompassing both continuous and discrete-time linear systems.
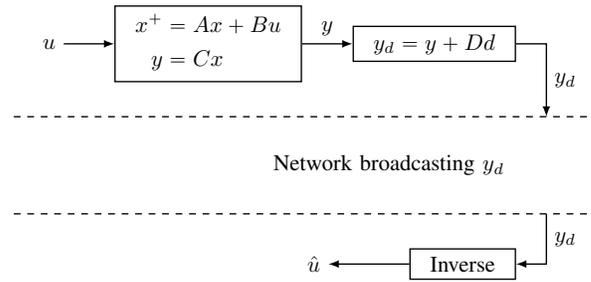
\begin{figure}[t]
    \centering
    \scalebox{0.85}{\input{fig/sche.tikz}}
    \caption{Schematic diagram of our considered setup: the upper part presents the exogenous input, the system, and the obfuscation strategy; the middle section represents the network over which $y_d$ is broadcast, with the eavesdropper; the lower part of the diagram illustrates the actions taken by a legitimate agent to reconstruct an estimate of $u$, denoted by $\hat{u}$.}
    \label{fig:sche}
\end{figure}

As mentioned, output redundancy is naturally associated with the null space of a rank-deficient output map, where it is often assumed that the rank of the output map equals the dimension of the input space. However, as mentioned in \cite{corona2016some}, this characterization does not encompass all potential scenarios, as it disregards epic (i.e., surjective) output operators with rank exceeding the dimension of the input space. This latter scenario signifies an intrinsic redundancy within the system, where only a proper subset of outputs is sufficient to uniquely determine the input trajectory. 
%
In this paper, we employ geometric tools to provide a comprehensive analysis of output-redundant LTI systems and delve into their properties through  suitable decomposition of both the state space and the output space.
Subsequently, we rigorously formulate the problem that we introduced at the beginning of the article and solve it by leveraging the properties of the output-redundant system that we derived. It is worth noticing that the findings presented here may have broader applications beyond the scope of the considered privacy preservation problem, for instance to fault detection and isolation~\cite{yang2022joint}.

Unsurprisingly, the problem dealt with in the paper can be regarded as a dual  of the input allocation problem formulated in~\cite{zaccarian2009dynamic}, where input redundancy (defined as properties of the system transfer matrix) is exploited to optimize additional metrics defined on the system input, without impacting output tracking performance. State-space input allocation strategies are also developed in \cite{serrani2012output} by leveraging a geometric analysis, which provides an alternative point of view to input redundancy. In our dual settings, however, we shift focus to an observation problem, treating the estimated state as  regulated variables influenced by a disturbance that has been intentionally injected (i.e., {\em allocated}) into the system output, hence affecting the state of a Luenberger-type observer. By dynamically allocating redundant measurements, we aim to impair the eavesdropper’s ability to reconstruct the system’s input. We refer to this strategy as the allocating mechanism or simply the allocator. The underlying principles of this duality are visually summarized in Table~\ref{fig:table_duality}. It should be noted that a more general definition of input redundancy and related taxonomy  (alternative to those of~\cite{zaccarian2009dynamic} and~\cite{serrani2012output}) that disposes of the assumption or right-invertibility has been recently proposed in \cite{Kreiss2021}. As systems invertibility is central to our discussion, a dualization of the more general definitions proposed therein has not been pursued in this work. 


%

\subsection{Contributions and outline of the paper}
The main contribution of this paper is as follows:
\begin{enumerate}
    \item We demonstrate and substantiate the necessity of output redundancy in a system to solve the problem of concealing the system's input while preserving the ability of a legitimate user to reconstruct it, as stated in the introduction.
    \item We provide a more comprehensive approach to output redundancy for LTI systems, where we formally distinguish between a strong and weak characterization in a geometric sense. This is a step beyond the frameworks in \cite{yang2019sensor,ren2023effects,cristofaro2024adaptive}, which are limited to a notion of strong output redundancy.
    \item We explore the properties of both strongly and weakly output-redundant systems under a geometric approach framework, which applies to various types of systems. Moreover, the properties are coordinate-free due to the merit of the geometric approach.
    \item We illustrate how the derived properties are leveraged to solve the problem of preventing the output eavesdropper from reconstructing the true system's input, while maintaining reconstructability of this latter by a legitimate user. 
\end{enumerate}

The paper is organized as follows. In Section~\ref{sec:pre}, we state the notation and provide essential concepts and notions of the geometric approach used in our work. We formalize the problem in Section~\ref{sec:p>m}, and prove the necessity of output redundancy for its solvability. In Section~\ref{sec:SOR} and Section~\ref{sec:WOR}, we present the properties and solutions to the synthesis problem by demonstrating the utilization of the derived properties of strong and weak output-redundant systems, respectively. Section~\ref{sec:simu} showcases simulation results validating the solutions and analysis in Section~\ref{sec:WOR}. Finally, concluding remarks are offered in Section~\ref{sec:end}.

\section{Notation and Preliminaries}
\label{sec:pre}
\subsection{Notation}
Throughout this paper, we use $\mathtt A$ to denote the matrix representation of a linear map $A:\X\to\Y$ among finite-dimensional vector spaces $\X$, $\Y$ over the field $\R$, i.e., $ \mathrm{Mat} (A) = \mathtt A$. For a finite-dimensional vector space $\X$, $\X^{\prime}$ denotes its dual space, i.e., the space of linear functionals $x^{\prime}:\X\to\R$. For a map $A:\X\to\Y$, $A^{\prime}:\Y^{\prime}\to\X^{\prime}$ denotes the dual map satisfying $\mathrm{Mat}(A^{\prime})=\mathtt{A}^{\top}$.  The external (internal) direct sum of vector spaces $\V$, $\W$ (subspaces $\V,\W\subseteq \X$) is denoted by $\V\oplus\W$, with elements $v\oplus w$, where $v\in\V$ and $w\in\W$~\cite[Section~0.3]{wonham1985linear}. The symbols $I$ and $\mathtt I$ stand respectively for the identity map and the identity matrix of compatible dimension, whereas 
$\mathbf{\mathtt 0}$ is an all-zeros matrix with compatible dimension.
For a symmetric square matrix $\mathtt M\in\R^{n\times n}$, $\mathtt M \succ (\prec)\ 0$ denotes  positive (negative) definiteness.
$\mathtt M^\dagger$ represents the pseudo inverse of $\mathtt M$. 
An LTI system is usually represented as $\{C,A,B\}$ with output map $C$, state map $A$, and input map $B$, respectively.
For an endomorphism $A: \X \to \X$, $\sigma(A)$ and $\sigma(\mathtt A)$ stand for the spectrum of the map $A$ and the eigenvalues of the matrix~$\mathtt A$, respectively. 
To streamline the presentation, maps that are instrumental for the completion of commutative diagrams will be defined in the caption of the figures presenting the diagrams themselves. These maps may also be referred to in the text.
\subsection{Preliminaries on Geometric Approach}\label{subsec:prelim}
This preliminary section follows both notation and concepts provided in \cite{wonham1985linear} and~\cite{massoumnia1986geometric}.
Let $C:\X \to \Y$ be a map; then, $C$ is an \textit{epimorphism} (or $C$ is \textit{epic}, onto) if $\im C = \Y$, and  $C$ is a \textit{monomorphism} (or $C$ is \textit{monic}, one-to-one) if $\ke C =0$. We say two finite-dimensional linear spaces $\X_1$ and $\X_2$ are \textit{isomorphic},  $\X_1 \simeq \X_2$, if they are dimensionally equal. If $\Ss \subseteq \Y$, $C^{-1}\Ss$ denotes the \textit{pre-image} of $\Ss$ under $C$. 
Let $C:\X \to \Y$, and let $\V \subset \X$ be a subspace with {\em insertion map} $V:\V \to \X$, i.e., $\V=\im{V}$ and $V$ is monic. The \textit{domain restriction} of $C$ to $\V$ is denoted by  $C|\V = C V$. 
Moreover, suppose $\im C \subseteq \W \subseteq \Y$. The \textit{codomain restriction} of $C$ from $\W$, denoted as $\W|C$ satisfies $W(\W|C)=C$, where
$W: \W \to \Y$ is the insertion map of $\W$.  

A subspace $\V\subseteq\X$ is said to be {\em invariant} with respect to a map $A:\X\to \X$ if $A\V\subseteq\V$. 
For an invariant subspace $\V$, we denote by $A|\V:\V\to\V$ the {\em restriction of $A$ to $\V$}, i.e., the unique map satisfying $AV=V(A|\V)$, where $V:\V \to \X$ is the insertion map of $\V$. 
Furthermore, we denote by $A|\X/\V$ or simply by  $\bar A:\X/\V \to \X/\V$ the map induced on $\X/\V$ by $A$, that is, the unique map satisfying $\bar{A}P=PA$, where $P:\X\to\X/\V$ is the canonical projection on $\X/\V$, the factor space $\X$ modulo $\V$. 
For a map $A:\X\to \X$ and subspaces  $\B\subseteq \X$, $\K \subseteq \X$, we define the smallest $A$-invariant subspace that contains $\B$ as $\ang{A\,|\, \B}$ and the largest $A$-invariant subspace that is contained in $\K$ as $\ang{\K\,|\, A}$ \cite[Section~2.1]{massoumnia1986geometric}.  

Let $A:\X\to \X$, $B:\U\to \X$ and $C:\X\to \Y$. We say a subspace $\V \subseteq \X$ is \textit{$(A,B)$-invariant} if there exists a map $F:\X\to \U$ such that $ (A+BF)\V \subseteq \V$. If this is the case, we say that the state-feedback map $F$ is a friend of $\V$, and denote the class of all friends of $\V$ by $\mathbf{F}(\V)$.
We say a subspace $\W \subseteq \X$ is \textit{$(C,A)$-invariant} if there exists a map $L:\Y\to \X$ such that $ (A+LC)\W \subseteq \W$. If this is the case, we also say that the output-injection map $L$ is a friend of $\W$, and denote the class of all friends of $\W$ by $\mathbf{L}(\W)$.
The notation $L\in \mathbf{L}(\W)$ reads ``$L$ is a friend of $\W$". 
Moreover, let $\B,\K\subseteq \X$.  We write $\underline{\V}(\K)$ to denote the class of $(A,B)$-invariant subspaces \textit{contained in} $\K$, and by $\underline{\W}(\B)$ the class of $(C,A)$-invariant subspaces \textit{containing} $\B$. The supremal and infimal elements of $\underline{\V}(\K)$ and $\underline{\W}(\B)$  are denoted by $\V^*(A,B;\K)$ and $\W^*(C,A;\B)$, respectively, or simply by $\V^*$ and $\W^*$~\cite[Section~4.2]{wonham1985linear}, \cite[Section~2.2]{massoumnia1986geometric}.

We say a subspace $\Rs \subseteq \X$ is a \textit{controllability subspace} if  $ \Rs = \ang{A+BF | \im{BG}}$ for some state feedback map $F:\X \to \U$ and input selection map $G:\U \to \U$. 
We say a subspace $\Ss \subseteq \X$ is an \textit{unobservability subspace} if $ \Ss = \ang{\ke HC\,|\,A+LC }$ for some output-injection map $L:\Y \to \X$ and measurement mixing map $H:\Y \to \Y$.
For $\B,\K\subseteq \X$, we write 
$\underline{\Rs}(\K)$ to denote the class of controllability subspaces \textit{contained in} $\K$, and
$\underline{\Ss}(\B)$ the class of unobservability subspaces \textit{containing} $\B$. 
The supremal and infimal elements of $\underline{\Rs}(\K)$ and $\underline{\Ss}(\B)$  are denoted by $\Rs^*(A,B;\K)$ and $\Ss^*(C,A;\B)$, respectively, or simply by $\Rs^*$ and $\Ss^*$, respectively~\cite[Section~5.1]{wonham1985linear}, \cite[Section~2.3]{massoumnia1986geometric}, \cite{basile1992controlled}.
\section{Problem formalization and the necessity of output redundancy} \label{sec:p>m}
We recall that the motivating problem described in the Introduction consists in injecting a disturbance signal on the output measurements of an LTI system to provide obfuscation of the input against eavesdroppers, whilst ensuring that a legitimate user can reconstruct the input from the perturbed output. This latter aspect of the problem is concerned with left inversion of a system affected by an unmeasurable output disturbance. Specifically, consider the system
\begin{subequations}\label{eq:sys}
    \begin{align}
        x^+ &= Ax + Bu, \quad x(0)=x_{0}\label{eq:sys:x}
        \\
        y_d &= Cx + Dd\label{eq:sys:y_d}
    \end{align}
\end{subequations}
where $x,x_{0} \in \X \simeq \R^n$, $u \in \U \simeq \R^m$, $d\in\D\simeq\R^{q}$, $y_d \in \Y \simeq \R^p$ are the system states, input, output disturbance and perturbed output, respectively; $A : \X \to \X, B:\U\to \X$, $D:\D\to\Y$ and $C: \X \to \Y$ are the corresponding maps.
For reasons that will become clear in the sequel, we define the aggregate input $u_{a}:=u\oplus d\in\U\oplus\D=:\U_{a}$, where $\U\oplus\D$ denotes the external direct sum of $\U$ and $\D$. 

We make the following standing assumptions on system~\eqref{eq:sys}:
\begin{assum}\label{as:obsv}
    The pair $(A,B)$ is controllable, and the pair $(C,A)$ is observable.
\end{assum}
\begin{assum}\label{as:li}
   The triplet $(C,A,B)$ is left-invertible (see~\cite{morse1971status},\cite[Section~4.3]{basile1992controlled} for definitions and properties.)
\end{assum}
\begin{rem}
Assumption~\ref{as:obsv} is made for the sake of simplicity, and can be removed at the cost of undue (and unnecessary) complications. On the other hand,  Assumption~\ref{as:li}, ensures reconstructibility (generally, by means of a non-causal system) of the input signal when $d=0$, which is clearly a necessary condition in the considered setup~\cite{Silverman1969}. 

{\em In the statements of the forthcoming results, we shall refrain from  mentioning explicitly  that  Assumptions~\ref{as:obsv} and~\ref{as:li} hold for~\eqref{eq:sys}, as these assumptions will henceforth be considered valid with no exception.}
\end{rem}

In order to formalize the problem dealt with in this paper, we introduce a definition of left-invertibility for system~\eqref{eq:sys} that  explicitly considers the presence of the disturbance:
%
\begin{defn}\label{def:li_d}
    Consider system~\eqref{eq:sys} with initial condition $x_0 \in \X$, input signal $u(\cdot) \in \mathbb U$, disturbance signal $d(\cdot)\in \mathbb D$, and output signal $y_{d}(\cdot) \in \mathbb Y$, where $\mathbb U$, $\mathbb D$, and $\mathbb Y$ are suitably defined spaces of functions defined over a given  interval $[0,T]$, $T>0$. Define 
    \begin{equation}\label{eq:operator}
        \theta_{x_0}(u,d) : \mathbb{U} \times \mathbb D \to \mathbb Y
    \end{equation}
    as the operator\footnote{For the sake of simplicity, we have used (and will henceforth use) the more economical expression~\eqref{eq:operator} in place of $  \theta_{x_0}(u(\cdot),d(\cdot))$.} that maps signals $u(\cdot), d(\cdot)$ to outputs $y_{d}(\cdot)$ over the interval $[0,T]$, parameterized by the initial condition $x_{0}$.  For example, for continuous-time systems,
    \[
     \theta_{x_0}(u,d)(t) =Ce^{At}x_{0}+\int_{0}^{t}Ce^{A(t-\tau)}Bu(\tau)d\tau+Dd(t)
    \]    
    for $t\in[0,T]$.
    System~\eqref{eq:sys} is said to be {left-invertible with respect to $u$ under the output disturbance $d$} if for any $x_0\in\X$ 
    \begin{equation}\label{eq:li_d}
        \theta_{x_0}(u_1,d_1) = \theta_{x_0}(u_2,d_2)
    \end{equation}
    implies
    \begin{equation}\label{eq:li_d:condU}
         u_1(t) = u_2(t)\quad \forall t \in [0,T]
    \end{equation}
\end{defn}
\smallskip

    Note that the above definition is not equivalent to  left-invertibility with respect to both $u$ and $d$, as it is not necessary for the condition
    $\theta_{x_0}(0,d_1) = \theta_{x_0}(0,d_2)$ to imply that $d_1(\cdot) = d_2(\cdot)$. However, left-invertibility with respect to the aggregate input $u_{a}=u\oplus d$ trivially implies left-invertibility with respect to $u$ under the disturbance $d$. Equivalent definitions of left-invertibility of systems defined by quadruplets $\{C,A,B,D\}$ are provided in~\cite[Chapter~8]{Trentelman2001}.
\begin{rem}
It should be clear that, due to linearity, system~\eqref{eq:sys} is  left-invertible with respect to $u$ under the output disturbance $d$ if and only if for $x_0=0$ 
    \begin{equation}\label{eq:li_d.bis}
        \theta_{0}(u,d) = \theta_{0}(0,0)
    \end{equation}
    implies $ u(t) = 0$ for all $t \in [0,T]$.
This equivalent formulation of Definition~\ref{def:li_d} will often be used in the sequel.
\end{rem}
\smallskip

The problem dealt with in this work is then stated as follows:

\begin{problem}\label{prob:li_d}
 For system~\eqref{eq:sys}, find $\D \neq 0$ and the map $D: \D \to \Y$ such that system~\eqref{eq:sys} is left-invertible with respect to $u$ under the output disturbance $d$, according to Definition~\ref{def:li_d}.
\end{problem}

\begin{rem}
    We stress that solvability of Problem~\ref{prob:li_d} is necessary to solve the motivating problem defined in the introduction, as a legitimate user must always be able to reconstruct $u(\cdot)$ when the output signal is intentionally corrupted. Conversely, the question of whether privacy can be guaranteed against an eavesdropping attacker for a given choice of the space $\D$ and the map $D$, must be evaluated separately, as it depends on the information available to the attacker.
\end{rem}

Up to this point, no specific properties have been postulated for the map $D$ or the space $\D$. However, the defense strategy should avoid the occurrence of disturbance signals producing no effect on the output, that is, the possibility that $D$ has a non-trivial null space. As a result,  it makes sense to assume that $D$ be monic, and, consequently, $\dim{\D}\leq p$.  This choice has an important implication on the property stated in Definition~\ref{def:li_d}:
\begin{prop}\label{prop:equiv-li}
Assume that $D$ is monic. Then, system~\eqref{eq:sys} is left-invertible with respect to $u$ under the output disturbance $d$ if and only if it is left-invertible with respect to the aggregate input $u_{a}=u\oplus d$.
\end{prop}
\begin{proof}
We shall prove necessity, as sufficiency is obvious. Let $u(\cdot)\in\mathbb{U}$, $d(\cdot)\in\mathbb{D}$ be such that~\eqref{eq:li_d.bis} hold. By virtue of the assumption of left-invertibility with respect to $u$ under $d$, it follows that $u(t)=0$ for all $t\in [0,T]$, and, consequently
\begin{equation}\label{eq:monic.D}
Dd(t)=0\qquad \forall\, t\in [0,T]
\end{equation}
As $D$ is monic,~\eqref{eq:monic.D} implies $d(\cdot)=0$, from which left-invertibility of the system with respect to $u_{a}$ follows.
\end{proof}

Once we have restricted $D$ to be monic, henceforth we shall identify $\D$ with a subspace of $\Y$, $\D\subseteq\Y$, and let $D:\D\to\Y$ denote the insertion map of $\D$ in $\Y$ (which is monic by definition). As a result, one disposes of the need to define $\D$ and $D$ separately. In the light of these observations and the result of Proposition~\ref{prop:equiv-li}, Problem~\ref{prob:li_d} is replaced by the following:
\begin{problem}\label{prob:li_d2}
 For system~\eqref{eq:sys}, find a non-trivial subspace $\D\subseteq\Y$ (with associated insertion map $D:\D\to\Y$) such that system~\eqref{eq:sys} is left-invertible with respect to the aggregate input $u_{a}=u\oplus d$.
\end{problem}
A somewhat obvious necessary condition for the solvability of Problem~\ref{prob:li_d2} follows:
\begin{thm}\label{thm:nec:OR}
    To solve Problem~\ref{prob:li_d2}, it is necessary that system~\eqref{eq:sys} be {\em output-redundant}, that is, $p=\dim\Y > \dim\U=m$. Furthermore, $q=\dim\D\leq  p-m$.
\end{thm}
\begin{proof}
As $\dim\U_{a}=m+q$, a necessary condition for left-invertibility with respect to $u_{a}$ is $m+q\leq p$, from which the statements of the theorem follow from the fact that $q\geq 1$ by assumption.
\end{proof}
Note that, for an output-redundant system,  Assumption~\ref{as:li} implies that $\dim(\im C) \geq m$. It is precisely the dimension of $\im C$ that prompts the following taxonomy:
\begin{defn}\label{def:SOR}
    System~\eqref{eq:sys} is said to be \textit{strongly output-redundant} (SOR) if it is output-redundant and $\dim(\im C) = m$.
\end{defn}
\begin{defn}\label{def:WOR}
    System~\eqref{eq:sys} is termed \textit{weakly output-redundant} (WOR) if it is output-redundant and $\dim(\im C)= p$.
\end{defn}

The relation between the dimension of $\U, \Y$ and the image of the output map  is illustrated in Figure~\ref{fig:in-out_number} for both the SOR and WOR cases. Note that it may be possible for system~\eqref{eq:sys} to satisfy $m < \dim(\im C) < p$. In this case, the system is said to be {\em generically output-redundant} (GOR). It should be clear that one need only focus on SOR and WOR systems to determine the properties of GOR systems as well.
\begin{figure}[t]
    \centering
    \subcaptionbox{SOR case}{\scalebox{0.49}{\input{fig/strong_number.tikz}}}
    \subcaptionbox{WOR case}{\scalebox{0.49}{\input{fig/weak_number.tikz}}}
    \caption{Relation between the dimension of input space, image of output map, and output space for: (a) strong output redundancy (SOR);  (b) weak output redundancy (WOR). }
    \label{fig:in-out_number}
\end{figure}
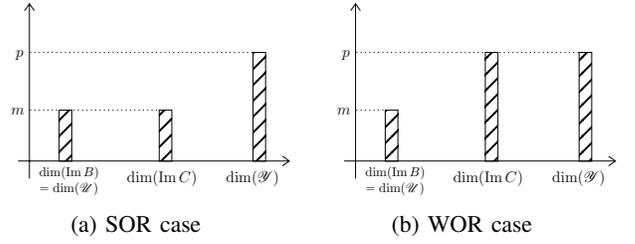
\section{Strong Output Redundancy}\label{sec:SOR}
In this section, we first introduce properties of interest for SOR systems, and then we illustrate the solution to the problem formalized in Section~\ref{sec:p>m} for this specific class.
\subsection{Properties of SOR Systems}\label{sec:SOR:prop}
Let system~\eqref{eq:system} be SOR. According to Definition~\ref{def:SOR}, the map~$C$ satisfies
\begin{equation}\label{eq:SORdim}
    p= \dim (\Y)  > \dim (\im C) = m,
\end{equation}
which implies that $C$ is not epic.
Consequently, it is possible to factor out $\im C$ from $\Y$ by canonical projection. The commutative diagram in Figure~\ref{fig:strong} depicts the factorization of $\Y$ into $\im C$ and $\Y/\im C$. As $\ke \pi = \im C$, the combined map $\pi C: \X\to\Y / \im C$ is the zero map. This in turn implies that no input $u(\cdot) \in \U$ can be retrieved from $\Y / \im C$; this simple observation will prove instrumental to the design of $\D$ solving Problem~\ref{prob:li_d2}.
\subsection{Design of $\D$ for SOR systems}
Owing to the aforementioned properties of SOR systems, the following result provides a necessary and sufficient condition on $\D$ for the solution to Problem~\ref{prob:li_d2}:
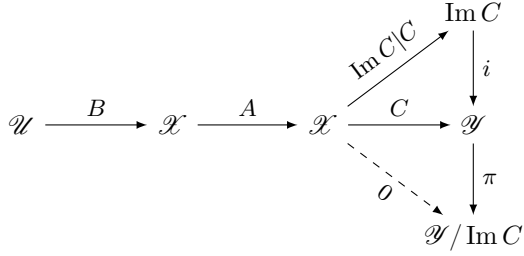
\begin{figure}[t]
    \centering
    \scalebox{1}{\input{fig/strong.tikz}}
    \caption{Commutative diagram for strongly output-redundant systems, where $\im C|C : \X \to \im C$ is the codomain restriction of $C$ to $\im C$, $i : \im C \to \Y$ is the insertion map, and $\pi : \Y \to \Y/\im C$ is the canonical projection $\mathrm{mod} \im{C}$.}
    \label{fig:strong}
\end{figure}
\begin{thm}\label{thm:SOR}
    Assume that system~\eqref{eq:sys} is SOR.
   A non-trivial subspace $\D\subseteq\Y$ solves Problem~\ref{prob:li_d2} if and only if 
    \begin{equation}\label{eq:SOR:iff}
        \D \cap \im C = 0.
    \end{equation}
\end{thm}
\smallskip
\begin{proof}
Owing to \eqref{eq:SORdim}, $\Y$ can be decomposed as 
\begin{equation}\label{eq:W}
\Y = \im C\, \oplus\, \W
\end{equation}
 where $\W \simeq \Y / \im C$ is an arbitrary complementary subspace to $\im C$.  Let $Q:\Y\to\Y$ be the projection on $\im C$ along $\W$, and $R:=I-Q$ be the projection on $\W$ along $\im C$.  Also, let $\tilde{Q}:=\im C|Q$ and $\tilde{R}=\W|R$ denote the corresponding natural projections. An example of the ensuing commutative diagrams is shown in Figure~\ref{fig:thm.SOR}. Accordingly, $y_{d}\in\Y$ is expressed uniquely as
\begin{equation}\label{eq:decomp.yd}
y_{d} = \tilde{Q}y_{d} \oplus \tilde{R}y_{d} = (\tilde{C}x+\tilde{D}_{1}d) \oplus \tilde{D}_{2}d
\end{equation}
where $\tilde{C}:=\tilde{Q}C$, $\tilde{D}_{1}:=\tilde{Q}D$, and $\tilde{D}_{2}:=\tilde{R}D$ (note that $\tilde{R}C=0$). We prove sufficiency by construction.  
Select a non-trivial subspace $\D \subseteq \Y$ satisfying~\eqref{eq:SOR:iff} and let the subspace $\W$ in~\eqref{eq:W} satisfy $\D\subseteq\W$.  Since~\eqref{eq:SOR:iff} implies $\tilde{D}_{1}=0$, from~\eqref{eq:decomp.yd} one obtains
\begin{equation}\label{eq:yd1}
y_{d} = \tilde{C}x \oplus \tilde{D}_{2}d.
\end{equation}
We will show that the system defined by the triplet $\{\tilde{C},A,B\}$ is left-invertible, and that the map $\tilde{D}_{2}$ is monic, which proves that~\eqref{eq:sys} is left-invertible with respect to $u_{a}$. By Assumption~\ref{as:li}, it follows that $B$ is monic and $\im{B}\cap\V^{*}(A,B;\ke{C})=0$ (see~\cite{morse1971status}). Consequently, to show that  $\{\tilde{C},A,B\}$ is left-invertible it suffices to prove that  $\im{B}\cap\V^{*}(A,B;\ke{\tilde{C}})=0$. To this end, note that
\[
\ke{\tilde{C}} = C^{-1}\big(\ke{\tilde{Q}}\cap\im{C}\big)=C^{-1}\left(\W\cap\im{C}\right)=\ke{C}
\]
hence $\V^{*}(A,B;\ke{\tilde{C}})=\V^{*}(A,B;\ke{C})$, from which left-invertibility of $\{\tilde{C},A,B\}$  follows. The assertion that $\tilde{D}_{2}$ is monic is proved in a similar fashion by noticing that  
\[
\ke{\tilde{D}_{2}}=D^{-1}(\ke{\tilde{R}}\cap\D)=D^{-1}(\im{C}\cap\D)=\ke{D}=0.
\]
Necessity is proven by contradiction. Assume that \eqref{eq:SOR:iff} does not hold and system~\eqref{eq:sys} is left-invertible with respect to $u_{a}$. Let $(u(\cdot),d(\cdot))\in\mathbb{U}\times \mathbb{D}$ be such that $\theta_{0}(u,d)=\theta_{0}(0,0)$. This implies $\tilde{D}_{2}d(t)=0$ for all $t\in[0,T]$, hence $d(t)\in D^{-1}\ke{\tilde{R}}\implies d(t)\in D^{-1}\im{C}\implies d(t)\in\D\cap\im{C}$ (seen as a subspace of $\D$) for all $t\in[0,T]$. Consequently, one can restrict the co-domain of $d(\cdot)$ to $\D\cap\im{C}$ whenever $\theta_{0}(u,d)=\theta_{0}(0,0)$. Accordingly, let $\tilde{D}_{3}:=\tilde{D}_{1}|(\D\cap\im{C})$ denote the domain restriction of $\tilde{D}_{1}=\tilde{Q}D$ to $\D\cap\im{C}$. Clearly, $\tilde{D}_{3}\neq 0$. By inspection, system~\eqref{eq:sys} is left-invertible with respect to $u_{a}\in\U\oplus\D$ only if the system
    \begin{align}\label{eq:sys.tilde}
        x^+ &= Ax + Bu, \quad x(0)=x_{0} \nonumber \\
        \tilde{y}_d &= \tilde{C}x + \tilde{D}_{3}\tilde{d},
    \end{align}
with output $\tilde{y}_{p}\in\im{C}$ and output disturbance $\tilde{d}\in\D\cap\im{C}$, is left-invertible with respect to the aggregate input $\tilde{u}_{a}=(u,\tilde{d})\in\U\oplus(\D\cap\im{C})$. This is an impossibility, given that 
\[
\dim \big(\U\oplus(\D\cap\im{C})\big) = m +\dim (\D\cap\im{C}) > m = \im{C}.
\]
\end{proof}
%
%
\begin{rem}[Maximum Coverage]
 The \textit{maximum coverage of $\Y$ from $\D$}, i.e., the subspace with largest dimension solving Problem~\ref{prob:li_d2}, is achieved when $\D \simeq \Y /\im C$.
\end{rem}

It is worth noticing that the necessary and sufficient condition given in Theorem~\ref{thm:SOR} relies solely on the definition of $\Y$ and $\im C$, and is not related to the mappings from the input space $\U$ to the state space $\X$.  As a result, for SOR systems,  the dynamics of~\eqref{eq:sys} plays no role in the solution of Problem~\ref{prob:li_d2}.\smallskip
\subsubsection*{Strategy for reconstructing the input for SOR systems}
The proof of Theorem~\ref{thm:SOR} suggests a simple strategy to reconstruct $u(\cdot)$ from the perturbed output $y_{d}(\cdot)$. Once the legitimate user has knowledge of the complementary subspace $\W$, and $\D\subseteq \W$, the effect of the disturbance can be removed from the output altogether by projection on $\im{C}$ along $\W$. This allows reconstruction of~$u(\cdot)$ via left-inversion of system~\eqref{eq:system}. The ensuing scheme is depicted in Figure~\ref{fig:scheme.SOR}.
\begin{figure}[t]
    \centering
    \scalebox{1}{\input{fig/thm-SOR.tikz}}
    \caption{Example of a commutative diagram for system~\eqref{eq:sys} as discussed in the proof of Theorem~\ref{thm:SOR}.}
    \label{fig:thm.SOR}
\end{figure}
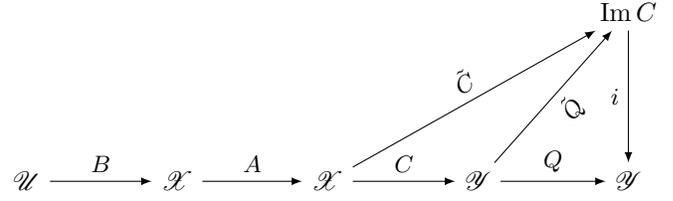
%

%
%
\begin{figure}[ht]
\centering
\includegraphics[width=0.9\linewidth]{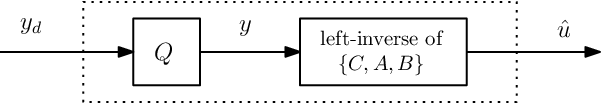}
       \caption{Strategy to reconstruct $u$ for SOR systems.}
        \label{fig:scheme.SOR}
\end{figure}
\subsection{Attacker's sensitivity to disturbance}
Recalling our motivating example, let us provide some insight into the sensitivity of an eavesdropper attack to the injected disturbance.
To this end, one must focus on a specific strategy to be implemented by the attacker. This is accomplished by defining an \textit{attacker model}, that is, the combination of system knowledge, resources, and objectives the attacker has at its disposal~\cite{teixeira2015secure}.
\begin{assum}[Attack model]
    \label{def:atk}
    Assume that  the attacker:
    \begin{enumerate}
        \item Knows the maps $(C,A,B)$ of system~\eqref{eq:system} {\em (system knowledge)};
        \item Has access to the corrupted output signal ${y}_d$ {\em (attacker's resources)};
        \item Attempts to reconstruct the input $u$ from observations of $y_d$ using, as a guess for $\D$, a subspace $\D_{a}$ providing maximum coverage (i.e., satisfying $\im C\oplus\D_{a}=\Y$), with associated projection $Q_a$ on $\im C$ along $\D_a$   {\em (attacker's objective)}.
    \end{enumerate}
\end{assum}
The following result follows immediately from~Theorem~\ref{thm:SOR}:
\begin{cor}\label{thm:SOR:nec}
    Assume that system~\eqref{eq:sys} is SOR. Suppose there is an eavesdropper attacker as in Assumption~\ref{def:atk}.
    The disturbance injection $d$ biases the attacker's reconstructed input only if 
    \begin{equation}\label{eq:DinDa}
        \D \not \subseteq \D_a.
    \end{equation}
\end{cor}
\begin{proof}
  The proof follows easily by contradiction.  Suppose \eqref{eq:DinDa} does not hold, i.e., that $\D \subseteq \D_a$.
    Since $\ke Q_a = \D_a$, it follows that $Q_a Dd = 0$, hence  $y_a \coloneqq Q_a y_d=y$. As system~\eqref{eq:sys} is left-invertible when $d=0$, the attacker can reconstruct the input $u$ without any bias.
\end{proof}

\begin{rem}
    One possible solution to achieving the attacker's objective, and possibly the most intuitive to use in a discrete time setting, is to exploit a least squares (LS) algorithm to compute an estimate of $u$, see, e.g., \cite{sain1969invertibility, valcher1999state, ansari2019deadbeat}. Indeed, for left-invertible systems (without the presence of an output disturbance), it has been shown that the input can be reconstructed exactly with some delay, depending on the relative degree of the system~\cite{ansari2019deadbeat}.
    For that solution strategy, the definition of the LS algorithm implies 
    \begin{equation}
        \D_a \coloneqq \ke C^{\prime}
    \end{equation}
\end{rem}
where $C^{\prime}:\Y^{\prime}\to\X^{\prime}$ is the dual map of $C$ and $\Y^{\prime},\X^{\prime}$ are dual spaces to $\Y,\X$, respectively.
%
%
In Section~\ref{sec:SOR:ex} we show an example of an SOR system, showcasing how satisfying \eqref{eq:DinDa} implies $Q_a D \neq 0$, although $Q_a C = C$.
%
%
\subsection{Illustrating example for SOR systems}\label{sec:SOR:ex}
We present a simple illustrating example concerning strongly output-redundant systems.
Consider system~\eqref{eq:sys} with $n=4$, $m = 2$, $p=3$ and map $C$ represented in given bases for $\X$ and $\Y$ as
\begin{equation*}
    \mathtt C \coloneqq \mathrm{Mat}(C) = \begin{bmatrix}
        1 & 1 & 0 & 0 \\
        1 & 0 & 0 & 0 \\
        2 & 1 & 0 & 0
    \end{bmatrix}.
\end{equation*}
Since $p=3 > \rank C =  m = 2$, the system is SOR.
The matrix representations of the maps $\im C|C$ and $i$ in the diagram of Figure~\ref{fig:strong} read as
\begin{align*}
   & \mathrm{Mat}(\im C|C) = \! \begin{bmatrix}
        1 & 0 & 0 & 0 \\
        0 & 1 & 0 & 0
    \end{bmatrix},\
    \mathrm{Mat}(i) = \! \begin{bmatrix}
        1 & 1 \\ 
        1 & 0 \\
        2 & 1
    \end{bmatrix}.
\end{align*}
As $p-\rank{C}=1$, the only possible nonzero subspace $\D$ satisfying~\eqref{eq:SOR:iff} is a complementary subspace to $\im{C}$ in $\Y$. As an arbitrary choice (among infinitely many) for $\D$, select the subspace with insertion map $D$ with coordinate representation
\begin{equation} \label{eq:W:SOR:eg}
    \mathtt D := \mathrm{Mat}(D) = \begin{bmatrix}
        0 & 1 & 0
    \end{bmatrix}^\top.
\end{equation}
Choose as basis of $\Y/\im{C}$ the cosets whose representatives are the elements of the basis of $\im{C}$. This yields the following representation for the canonical projection $\pi:\Y\to\Y/\im{C}$ in~Figure~\ref{fig:strong}:
\[
   \mathrm{Mat}(\pi) = \! \begin{bmatrix}
       1 & 1 & -1 
   \end{bmatrix}
\]
It can be easily verified that the projection $Q:\Y\to\Y$ on $\im{C}$ along $\D$ has representation
\begin{equation*}
   \mathtt Q = \begin{bmatrix}
        1 & 0 & 0 \\
        -1 & 0 & 1 \\
        0 & 0 & 1 
    \end{bmatrix}
\end{equation*}
yielding $\mathtt{QC = C}$ and $\mathtt{QD} = 0$, hence $\mathtt Q\bar y = \mathtt Cx$; this implies that Problem~\ref{prob:li_d2} can be solved.
Assuming that the attacker employs a least-squares algorithm, one obtains 
\[
\D_a= \ke {C}^\prime \simeq \ke \texttt{C}^\top = \mathrm{span}\left\{\begin{bmatrix}
    -1 & -1 & 1
\end{bmatrix}^\top\right\}
\] 
and, for the projection $Q_{a}:\Y\to\Y$ on $\im{C}$ along $\D_{a}$ (apart from round-off errors)
\[
\mathtt{Q}_{a}=
\begin{bmatrix}
   0.667  & -0.333  &  0.333 \\
  -0.333  &  0.667  &  0.333 \\
   0.333  &  0.333  &  0.667
 \end{bmatrix}
\]
Consequently, $\mathtt{Q}_{a}\mathtt{D}\neq0$, leading to a bias of the attacker's estimate.
%

%
%
%
\section{Weak Output Redundancy}\label{sec:WOR}
In this section, we first introduce the properties of WOR systems and develop an appropriate defense strategy against the malicious eavesdropper for this class of systems.
\subsection{Properties of WOR systems}\label{subsec:WOR-prop}
Assume that system~\eqref{eq:system} is WOR. Let $\Ss^*$ denote the infimal unobservability subspace containing $\im B$, i.e., $\Ss^* \coloneqq \Ss^*(C,A;\im B)$. The following (fundamental) result has been proved in~\cite{yang2022joint}:
\begin{lem}[\!\!\cite{yang2022joint}]\label{lem:dimSs*}
 For a WOR system,  $\Ss^*$ is non-trivial, and satisfies $m\leq \dim(\Ss^*) \leq n-p+m$.
\end{lem}
\smallskip
Since the class of friends of $\Ss^{*}$ is non empty for WOR systems, select $L \in \textbf{L}(\Ss^*)$ and consider a classical Luenberger observer for \eqref{eq:system}
\begin{equation} \label{eq:Luen}
    \hat{x}^+ = (A+LC) \hat x - Ly ,
\end{equation}
where $L: \Y \to \X$ is the output-injection map. Define the state and output observation errors as 
\begin{equation*}
    \tilde x \coloneqq x - \hat x , \qquad
    \tilde y \coloneqq y - C\hat x
\end{equation*}
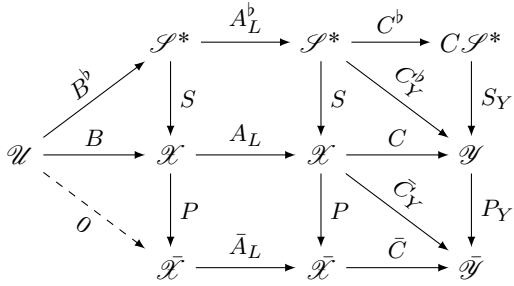
\begin{figure}[t]
    \centering
    \scalebox{1}{\input{fig/weak.tikz}}
    \caption{Commutative diagram for weakly output-redundant systems. In the diagram, $\bar{\X}\coloneqq \X/\Ss^*$ and $\bar{\Y}\coloneqq \Y/(C\Ss^*)$. $S:\Ss^{*}\to \X$ is the insertion map of $\Ss^{*}$ in $\X$, $S_Y: C\Ss^* \to \Y$ is the insertion map of $C\Ss^{*}$ in $\Y$, $P:\X\to\bar{\X}$ is the canonical projection modulo $\Ss^{*}$, and $P_{Y}:\Y\to\bar \Y$ is the canonical projection modulo $C\Ss^*$. 
    The codomain restriction $B^\flat:\U\to\Ss^{*}$ is well-defined since $\im B\subseteq\Ss^{*}$, and the map $\bar{C}:\bar \X \to \bar \Y$ is well-defined due to the fact that $\ke{\bar{C}_Y}=\ke{P_{Y}C}=\Ss^{*}+\ke C \supseteq \Ss^*$. The map $A_{L}^{\flat}:\Ss^{*}\to\Ss^{*}$ is the restriction of $A_{L}$ to the $A_{L}$-invariant subspace $\Ss^{*}$. The map $\bar{A}_{L}:\bar{\X}\to\bar{\X}$ is the map induced on $\bar{\X}$ by $A_{L}$.}
    \label{fig:weak}
\end{figure}
with associated {\em error system} 
\begin{equation}\label{eq:system-diff}
\begin{aligned}
    \tilde x^+ &= A_{L} \tilde x + Bu \\
    \tilde y &= C \tilde x.
\end{aligned}
\end{equation}
where $A_{L}:=A+LC$. {\em We regard \eqref{eq:system-diff} as a new system altogether, obtained from~\eqref{eq:system} via an output-injection transformation using a friend of $\Ss^{*}$.}
The maps defining system~\eqref{eq:system-diff} are then represented according to the commutative diagram in Figure~\ref{fig:weak}, from which one can define two distinct systems:
\begin{itemize}\addtolength{\parskip}{1mm}
    \item System $\Sigma^\flat \coloneqq \{C^\flat, A_L^\flat,B^\flat \}$ has input space $\U$, state space $\Ss^*$ and output space $C\Ss^{*}$.
    \item System $\bar \Sigma \coloneqq \{\bar C, \bar A_L,0\}$ is autonomous with respect to the input space $\U$, and has state space $\bar{\X}=\X/\Ss^{*}$ and output space~$\bar{\Y}=\Y/(C\Ss^{*})$.
\end{itemize}
Relevant properties of the systems $\Sigma^\flat$ and $\bar \Sigma$ are presented in the following lemmas:
\begin{lem}[\!\!\cite{massoumnia1986geometric}{ Proposition~13, Theorem~16}]\label{lem:both-obsv}
    Let system~\eqref{eq:system} be WOR. Assume that $(C,A)$ is observable. Then, the pairs $(C^\flat,A_L^\flat)$ and $(\bar C,\bar A_L)$ are observable for any $L \in \mathbf{L}(\Ss^*)$. 
\end{lem}
\smallskip
\begin{lem}[\!\!\cite{yang2022joint}]\label{lem:square-invert}
    Let system~\eqref{eq:system} be WOR. Assume that $\{C,A,B\}$ is left-invertible. Then, system $\Sigma^\flat$ is square and invertible. Consequently, $\dim\bar{\Y}=p-m$.
\end{lem}
\smallskip
\begin{rem}\label{rem:stable}
    When $(C,A)$ is observable, the spectra of $A_L^\flat$ and $\bar A_L$ can both be freely assigned by a friend of $\Ss^{*}$. Consequently, we shall always $L \in \textbf{L}(\Ss^*)$ such that $A_L^\flat$ and $\bar A_L$ are both stable. An algorithm for calculating such $L$ is provided in Appendix~\ref{app:calcL}.
\end{rem}
\smallskip

Under the assumption that the original system \eqref{eq:system} is WOR,  we have derived, via output-injection, an ``error system"~\eqref{eq:system-diff} where the unobservability subspace $\Ss^*$ has been rendered invariant. The error system is represented as the interconnection of the \textit{square and invertible} system $\Sigma^\flat$ with the \textit{autonomous} system~$\bar \Sigma$. The ensuing structure strategically confines all influences from the input space $\U$ through the map $B$ into $\Ss^*$, while the system $\bar \Sigma$ is decoupled from the input. This distinctive feature of WOR systems will be explored in the subsequent section, when the output disturbance $d$ is introduced. 
\subsection{Design of $\D$ for WOR systems}
We now consider the presence of the output disturbance~$d$ on a weakly output-redundant system. Let system~\eqref{eq:sys} by WOR, and apply an output-injection transformation with  $L\in\mathbf{L}(\Ss^{*})$.  This yields, in place of~\eqref{eq:system-diff}, the new error system
%
%
%
%
%
\begin{equation}\label{eq:diff-sys-Wd}
    \begin{split}
        \tilde x^+ &= A_{L} \tilde x + Bu + D_Ld \\
        \tilde{y}_{d} &= C \tilde x + Dd,
    \end{split}
\end{equation}
 where $D_L \coloneqq LD$ denotes the input disturbance map. It is noted that the output disturbance influences both the state and the output of the error system~\eqref{eq:diff-sys-Wd}. 

The next step is to define a preliminary condition on $\D$ with useful implications.
\begin{lem}\label{lem:W-design}
    Let system~\eqref{eq:sys} be WOR. Assume that  $\D \subseteq \Y$ satisfies
    \begin{equation}\label{eq:imD}
        \D \subseteq C\Ss^*.
    \end{equation}
    Then, the map $P_YD: \D \to \bar \Y$ is the zero map.
\end{lem}
\begin{proof}
    As $P_Y$ is the canonical projection  modulo $C\Ss^{*}$ on $\Y$,  $\ke P_Y = C\Ss^*$. Thus,  $P_Y D = 0$ follows directly from \eqref{eq:imD}.
\end{proof}
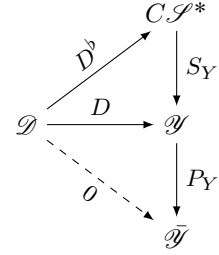
\begin{figure}
    \centering
    \scalebox{1}{\input{fig/injectD.tikz}}
    \caption{Decomposition of the map $D$ in Lemma~\ref{lem:W-design}. The codomain restriction $D^\flat=C\Ss^{*}|D$, which is well-defined since $\im D \subseteq C\Ss^*$ by assumption, is also the insertion map of $\D$ in $C\Ss^{*}$.}
    \label{fig:d-inject}
\end{figure}
Lemma~\ref{lem:W-design} is represented by the commutative diagram in Figure~\ref{fig:d-inject}, showing that $d$ does not affect~$\bar \Y$. 
%

Assuming that $\D \neq 0$ satisfies~\eqref{eq:imD}, applying the same decomposition as in~Figure~\ref{fig:weak} to the perturbed system~\eqref{eq:diff-sys-Wd} yields the following equations for the dynamics of the system defined on the factor space $\bar{\X}$:
%
\begin{equation}\label{eq:sigma.bar}
    \begin{aligned}
        \bar{x}^{+} &= \bar A_L \bar{x} + \bar D_L d \\
        \bar{y} &= \bar C \bar{x}
    \end{aligned}
\end{equation}
where $\bar{x}\in\bar{\X}$, $\bar{y}\in\bar{\Y}$, $d\in\D$, and the maps $\bar{A}$, $\bar{D}_{L}$ and $\bar{C}$ are defined in Figure~\ref{fig:d-sys-decompose}.
The main result of this section will be devoted to show that it is always possible to design $\D\neq 0$ satisfying~\eqref{eq:imD} such that system~\eqref{eq:sigma.bar} is left-invertible (with respect to the input  $d$). In turn, this result will lead to the solution of Problem~\ref{prob:li_d2} for system~\eqref{eq:diff-sys-Wd}  by way of reconstructability of $d(\cdot)$ from  $\tilde{y}(\cdot)$ via canonical projection modulo $C\Ss^{*}$.
%
\begin{figure}
    \centering
    \scalebox{1}{\input{fig/decomposeD.tikz}}
    \caption{Decomposition of system $\{C,A_{L},D_L\}$ resulting from $L\in\mathbf{L}(\Ss^{*})$, where $\bar D_L \coloneqq PLD$.}
    \label{fig:d-sys-decompose}
\end{figure}
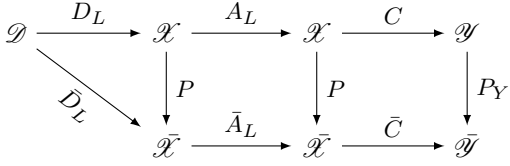
%

To this end, recall that system~\eqref{eq:sigma.bar} is left-invertible if and only if the following two conditions hold (see~\cite[Theorem~5]{morse1971status} and~\cite[Proposition 9]{massoumnia1986geometric}):\smallskip
%
%
\begin{enumerate}[label=\textit{(C\arabic*)}]
\addtolength{\parskip}{1mm}
    \item \label{Cond:monic} $\bar D_L: \D \to \bar \X$ is monic;
    \item \label{Cond:cap=0} $\im(\bar D_L) \cap \bar \V^* = 0$, where
    \begin{equation}\label{eq:def:barV*}
        \bar \V^*  \coloneqq  \V^* (\bar{A}_L, \bar D_L;\ke \bar C).
    \end{equation}
\end{enumerate}
%
%
\begin{figure}
    \centering
    \scalebox{1}{\input{fig/SDbarX.tikz}}
    \caption{Diagram showing the action of $D:\D\to\X$ on $\bar \X$ via its codomain restriction on $C\Ss^*$. Note that this diagram can be thought of as ``zooming in'' on the map $\bar D_L$ in Figure~\ref{fig:d-sys-decompose}.}
    \label{fig:S*-D-barX}
\end{figure}
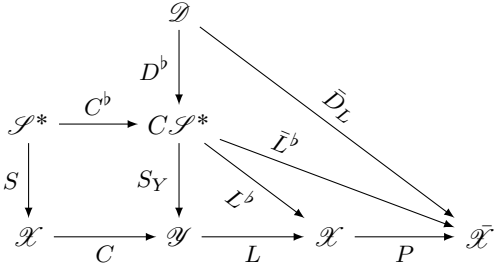
Using the diagram in Figure~\ref{fig:S*-D-barX}, it follows that 
\begin{equation*}
    \bar D_L \coloneqq PLD = \bar L^\flat D^\flat
\end{equation*}
where $\bar L^\flat: C\Ss^* \to \bar \X$ is defined as $\bar L^\flat \coloneqq PLS_Y$, the domain restriction of $\bar L \coloneqq PL$ on $C\Ss^{*}$. 
%
Before proceeding further with the design of $\D$, one must prove that $\bar{L}^\flat$ is a nonzero map, which is clearly  a necessary condition for condition~\ref{Cond:monic} above to hold.
%
%
\begin{prop}\label{prop:L_Z-nonzero}
    Let system~\eqref{eq:sys} be WOR, and let $L:\Y\to\X$ be a friend of $\Ss^*$. Then, $\bar{L}^\flat \neq 0$.
\end{prop}
%
%
\begin{proof}
    Let $L \in \mathbf{L}(\Ss^*)$. As  $\Ss^{*}$ is $A_{L}$-invariant, it follows that $\im{A_{L}S}\subseteq\Ss^{*}$, hence
    \begin{equation}\label{eq:friendL}
    0 = PA_{L}S =  P(A+LC)S = PAS + PLCS 
    \end{equation}
    Using the diagram in Figure~\ref{fig:S*-D-barX}, it is seen that
    \begin{equation}
        PLCS = \bar{L}^\flat C^\flat.
    \end{equation}
    Consequently,  $\bar{L}^\flat=0$ implies $PAS=0$, which is equivalent~to 
    \[
    A\Ss^{*}\subseteq \ke{P}=\Ss^{*}
    \]
    showing that $\bar{L}^\flat=0$ implies that $\Ss^{*}$ is $A$-invariant as well. Since $\Ss^{*}$ is a proper subset of $\X$ by virtue of Lemma~\ref{lem:dimSs*},  and $\Ss^{*}\supseteq\im{B}$ by definition, it follows that $\bar L^\flat = 0$  contradicts controllability of the pair $(A,B)$, thus violating Assumption~\ref{as:obsv}.
\end{proof}

We are now in the position to state and prove the main result of this section:
\begin{prop}\label{lem:leftInv_Sigma_d}
    Assume that system~\eqref{eq:sys} is WOR. For any $L \in \mathbf{L}(\Ss^*)$, there exists a non-trivial subspace $\D \subseteq  C\Ss^*$ such that system~\eqref{eq:sigma.bar} is left-invertible.
\end{prop}
\begin{proof}
Fix, arbitrarily, $L \in \mathbf{L}(\Ss^*)$. We first show that it is always possible to select $\D\subseteq C\Ss^{*}$ such that its insertion map $D^{\flat}:\D \to C\Ss^{*}$ yields a monic map $\bar D_L: \D \to \bar \X$, i.e., that condition~\ref{Cond:monic} is satisfied. To this end, select $\D$ as a complementary subspace of $\ke{\bar{L}}^{\flat}$ in $C\Ss^{*}$, that is, let $\D\subseteq\Y$ satisfy
    \begin{equation}\label{eq:assign-D}
        C\Ss^* = \ke \bar L^\flat \oplus \D.
    \end{equation}
This is possible by virtue of the fact that $\bar{L}^{\flat}\neq0$ implies $\dim(\im{\bar{L}^{\flat}})>0$, hence, by the first isomorphism theorem,
$\dim(C\Ss^{*}/\ke{\bar{L}^{\flat}})>0$. Since any complementary subspace of $\ke{\bar{L}}^{\flat}$ in $C\Ss^{*}$ is isomorphic to $C\Ss^{*}/\ke{\bar{L}^{\flat}}$, it follows that $\dim\D>0$ in~\eqref{eq:assign-D}. Since
\[
\ke{\bar{D}_{L}}=\ke(\bar{L}^{\flat}D^{\flat}) ={D^{\flat}}^{-1}\left(\ke{\bar{L}^{\flat}}\cap \im{D^{\flat}}\right)
\]
and $\ke{\bar{L}^{\flat}}\cap \D=0$ by~\eqref{eq:assign-D}, one obtains 
\[
\ke{\bar{D}_{L}}=\ke{D^{\flat}}=0
\]
since $D^{\flat}$ is monic by definition. Consequently, $\bar{D}_{L}$ is monic.

Next, we show that $\D$ can be chosen such that, in addition, condition~\ref{Cond:cap=0} holds. Clearly, if $\D$ selected at the previous step is such that  $\im(\bar D_L) \cap \bar \V^* = 0$, the proof is complete. Conversely, assume that this is not the case, that is, $\im(\bar D_L) \cap \bar \V^* \neq 0$ for the previous choice of $\D$ satisfying~\eqref{eq:assign-D}. We claim that, necessarily,
    \begin{equation}\label{eq:barD_L:no-subset}
        \im \bar D_L \not\subseteq \bar \V^*.
    \end{equation}
From the definition of $\bar \V^*$, it follows that
    \begin{equation*}
        \bar A_L \bar \V^* \subseteq \bar \V^* + \im \bar D_L = \bar \V^*,
    \end{equation*}
    which implies $\bar \V^*$ is also $\bar A_L$-invariant. By definition, $\bar \V^* \subseteq \ke \bar C$, which implies that $\bar \V^*\neq0$ is contained in the unobservable subspace of the pair $(\bar C,\bar A_L)$. This contradicts the fact that $(\bar C,\bar A_L)$ is observable, as asserted in Lemma~\ref{lem:both-obsv}. Therefore, we conclude that \eqref{eq:barD_L:no-subset} must hold. From \eqref{eq:barD_L:no-subset}, it is possible to define a non-zero subspace $\bar \N \subset \bar \X$ such that
    \begin{equation}\label{eq:barN}
        \im \bar D_L = \left( \bar \V^* \cap \im \bar D_L \right) \oplus \bar \N .
    \end{equation} 
    Define a new subspace $\D_1$ as
    \begin{equation*}
        \D_1 \coloneqq \left( \bar D_L \right)^{-1} \bar \N
    \end{equation*}
    and note that $\D_{1} \subseteq \D \subseteq C\Ss^*$. Let $D_1^\flat: \D_1 \to C\Ss^*$ denote the domain restriction of $D^{\flat}:\D\to C\Ss^{*}$ to $C\Ss^{*}$, and $\bar D^1_L: \D_1 \to \bar \X$ the domain restriction of $\bar{D}_{L}$ to   $\D_{1}$ (see Figure~\ref{fig:D1}).  Notice that $\bar D_L^1$ is still monic since $\bar D_L$ is monic and $\D_1 \subseteq \D$. Moreover, since
    \begin{equation*}
        \bar D_L^1 \D_1 = \bar D_L^1 \left( \bar D_L^{-1} \bar \N \right) = \bar \N
    \end{equation*}
    one obtains
    \begin{equation}\label{eq:ImbarD1L}
        \im \bar D_L^1 = \bar \N.
    \end{equation}
    Since $\bar \N \cap \bar \V^* = 0 $ by virtue of \eqref{eq:barN}, from~\eqref{eq:ImbarD1L} it follows that 
    \begin{equation}\label{eq:lem:part}
        \bar \V^* \cap \im \bar D_L^1 = 0.
    \end{equation}
    Note  that $\bar \V^*$ defined in \eqref{eq:def:barV*} is computed  using $\bar D_L$, not $\bar D_L^1$; consequently, \eqref{eq:lem:part} is not yet sufficient to conclude that system~\eqref{eq:sigma.bar} is left-invertible when $D_{L}^{1}$ replaces $D_{L}$ (alternatively, $\D_{1}$ replaces $\D$).
 To show that this is indeed the case, define
    \begin{equation*}
        \bar \V^*_1 \coloneqq \V^*(\bar{A}_L, \bar D_L^1; \ke \bar C)
    \end{equation*}
    and note that, by the properties of $(A,B)$-invariant subspaces,
    \begin{equation}\label{eq:prop:V*}
        \bar A_L \bar \V^*_1 \subseteq \bar \V_1^* + \im \bar D_L^1.
    \end{equation}
   Since, by virtue of \eqref{eq:barN} and \eqref{eq:ImbarD1L}
    \begin{equation*}
        \im \bar D_L^1 \subseteq \im \bar D_L ,
    \end{equation*}
    it follows from \eqref{eq:prop:V*} that
    \begin{equation*}
        \bar A_L \bar \V^*_1 \subseteq \bar \V_1^* + \im \bar D_L,
    \end{equation*}
    which implies that $\bar \V^*_1$ is also a $(\bar A_L,\bar D_L)$-invariant subspace contained in $\ke \bar C$. Since $\V^*$ is the supremal element of this class, one obtains
    \begin{equation*}
        \bar{\V}^*_1 \subseteq \bar{\V}^*,
    \end{equation*}
    which, according to \eqref{eq:lem:part}, yields
    \begin{equation*}
        \bar{\V}^*_1 \cap \im \bar{D}^L_1= 0.
    \end{equation*}
    As a result, left-invertibility of system~\eqref{eq:sigma.bar} follows from selecting $\D_{1}$ as the disturbance subspace. 
\end{proof}
    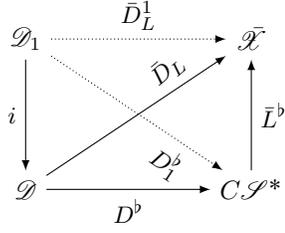
\begin{figure}
    \centering
    \scalebox{1}{\input{fig/D1.tikz}}
    \caption{Commutative diagrams illustrating the newly defined maps (in dotted arrows) for $\D_1 \coloneqq \left( \bar D_L \right)^{-1} \bar \N$. The  nondescript notation $i$ is used for the insertion map $i:\D_{1}\to\D$.}
    \label{fig:D1}
    \end{figure}
\smallskip
    
The following result, providing the solution of Problem~\ref{prob:li_d2} for WOR systems, is essentially a corollary of Proposition~\ref{lem:leftInv_Sigma_d}.
\begin{thm}\label{thm:WOR:li}
   Assume that system~\eqref{eq:sys} is WOR. Then, there exists a non-trivial subspace $\D \subseteq\Y$ solving Problem~\ref{prob:li_d2}.
\end{thm}
\begin{proof}
We begin by recalling that left-invertibility is invariant under output-injection transformations. Consequently, system~\eqref{eq:sys} is left-invertible with respect to the aggregate input~$u_{a}$ if and only if so is system~\eqref{eq:diff-sys-Wd}. Choose, arbitrarily, $L\in\mathbf{L}(\Ss^{*})$, and select $\D\subseteq C\Ss^{*}$ as in Proposition~\ref{lem:leftInv_Sigma_d}. Let $(u(\cdot),d(\cdot))\in\mathbb{U}\times \mathbb{D}$ be such that $\theta_{0}(u,d)=\theta_{0}(0,0)$ for system~\eqref{eq:diff-sys-Wd}. In particular, this implies that $\bar{y}(t)=P_{Y}\tilde{y}_{d}(t)=0$ for all $t\in [0,T]$. Since system~\eqref{eq:sigma.bar} is left-invertible by construction, it follows that $d(\cdot)=0$. Letting $d=0$ in~\eqref{eq:diff-sys-Wd}, one obtains $u(\cdot)=0$ by way of left-invertibility of the triplet $\{C,A_{L},B\}$. Consequently, system~\eqref{eq:diff-sys-Wd} is left-invertible with respect to $u_{a}=(u,d)$, and so is system~\eqref{eq:sys}.
\end{proof}

\begin{rem}
A numerical algorithm for the computation of the matrix representation of the insertion map~$D$, based on the method presented in the proof of Proposition~\ref{lem:leftInv_Sigma_d}, is provided in Appendix~\ref{app:calcD}. We stress that other solutions may exist, given the fact that Theorem~\ref{thm:WOR:li} provides a sufficient condition. It is indeed an open problem whether the proposed algorithm to find $\D$ defines the supremal element of all disturbance spaces solving Problem~\ref{prob:li_d2}.
\end{rem}
\subsubsection*{Strategy for reconstructing the input in WOR systems}
The strategy for producing a reconstruction $\hat{u}$ of $u$, which stems from the results presented in this section, is illustrated by the scheme in Figure~\ref{fig:scheme.WOR}. The observer~\eqref{eq:Luen} with injection gain $L\in\mathbf{L}(\Ss^{*})$ (selected such that the spectrum of $A_{L}$ is stable) is used to generate the ``output error'' $\tilde{y}_{d}$ of  system~\eqref{eq:diff-sys-Wd}. Projection modulo $C\Ss^{*}$ of $\tilde{y}_{d}$ yields the output $\bar{y}$ of system~\eqref{eq:sigma.bar}, which is used to reconstruct the disturbance $d(\cdot)$ by way of left-inversion of the triplet $\{\bar{C},\bar{A}_{L},\bar{D}_{L}\}$. The resulting signal $\hat{d}(\cdot)$ is used to reconstruct the output $y(\cdot)$ of the unperturbed  system~\eqref{eq:system} as $\hat{y}=y_{d}-D\hat{d}$. The signal $\hat{y}(\cdot)$, in turn, is used to generate $\hat{u}(\cdot)$ via left-inversion of system~\eqref{eq:system}. Alternatively, the lower-dimensional model $\Sigma^\flat=\{C^\flat,A_L^\flat,B^\flat\}$ can be used in place of the original system, as this option may be advantageous from a computational standpoint. Specifically, once the disturbance signal has been reconstructed, a recostruction of the output~$\tilde{y}$ of the error system~\eqref{eq:system-diff} is provided by $\hat{\tilde{y}}=\tilde{y}_d-D\hat{d}$. The signal $\hat{\tilde{y}}(\cdot)$ is then used to reconstruct $u(\cdot)$ by way of left-inversion of~$\Sigma^\flat$.
\begin{figure}[ht]
\centering
\includegraphics[width=0.99\linewidth]{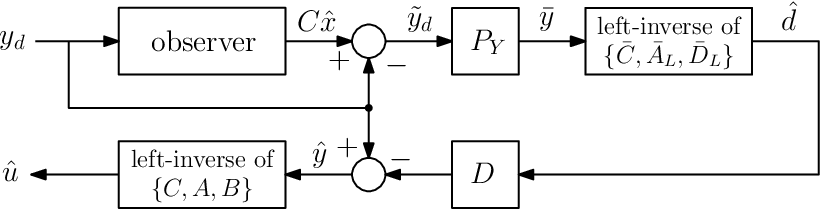}
       \caption{Strategy to reconstruct $u$ for WOR systems.}
        \label{fig:scheme.WOR}
\end{figure}
\subsection{Attacker's Sensitivity to Disturbance}
In the case of WOR systems, since  $\Y = \im C$, the only subspace $\D_a\subseteq\Y$ that can be chosen by the attacker such that $\Y = \im C \oplus \D_a$ is the trivial subspace, i.e., $\D_a = 0$. This implies that no projection methods can be used to eliminate the effect of the disturbance on the output. 
There is, however, one possible strategy for the attacker to correctly reconstruct the input $u$ \textit{if} it knows $\D$, as it can implement the same strategy as the defense. 
This class of attacker may be considered as ``subspace aware'', and strategies addressing this eventuality will be considered in future work. Further analysis of the obfuscation properties of our scheme, namely if it is necessary at all for the attacker to have knowledge of the exact $\D$ used by the defender, is also left for future work. In Section~\ref{sec:simu}, we show how our proposed strategy may indeed be a solution to the obfuscation problem, by considering an attacker using a least squares algorithm to reconstruct $u$ from $y_d$. 
%

%
%
%

\section{Simulation Example: WOR systems} \label{sec:simu}
\begin{figure*}[t]
    \centering
    \scalebox{0.85}{\input{fig/DGU.tikz}}
    \caption{Schematic diagram of the coupled distributed generation units (DGUs).}
    \label{fig:DGU}
\end{figure*}
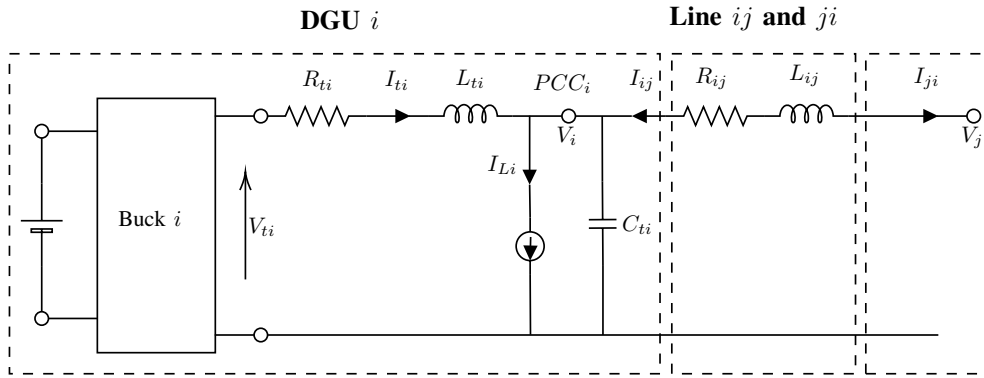
In this section, we provide an example for WOR systems, showing that Problem~\ref{prob:li_d} is indeed solved for our design of $\D$, as well as demonstrating how the disturbance injection can bias the estimate of the input by an eavesdropping attacker.
We consider an islanded DC microgrid composed of two coupled distributed generation units (DGUs), where buck converters are coupled via RLC filters (see Fig.~\ref{fig:DGU}). Each DGU can be modeled as the following dynamical system \cite{tucci2017line}:
\begin{equation}
    \label{eq:DGUi}
    \begin{split}
        C_i\dot V_i &= I_{ti} - I_{Li} + \sum_{j\in \mathcal N_i} \frac{V_j - V_i}{R_{ij}} \\
        L_{ti} \dot I_{ti} &= V_{ti} - V_i - R_{ti} I_{ti} 
    \end{split}
\end{equation}
where $V_i, I_{ti}$ are the voltage at the point of common connection and terminal current, $V_{ti}$ is the terminal voltage, and $I_{Li}$ is a load current.
The terminal voltage $V_{ti}$ is the control input, which is obtained via the decentralized PI controller defined in \cite{tucci2017line}, ensuring that $V_i$ asymptotically tracks an externally defined reference.
The parameters $C_i, R_{ti},L_{ti}$ are the capacitance, inductance and resistance of the DGU's RLC filter.
The system composed of two DGUs can be cast as~\eqref{eq:system} by taking the closed-loop dynamics from \cite{tucci2017line}, and $u = I_L = \col_{i\in\{1,2\}} I_{Li}$, and $y  = \col_{i\in \{1,2\}} \matrices{V_i & I_{ti}}^\top$. In this context, the estimation of $u$ can be used to achieve improved performance and system balancing by the system operator. However, an eavesdropping attacker may exploit estimated values of the load currents to infer users' personal habits and behaviors \cite{kumar2019smart}. 

We consider an attacker capable of storing $r \in \mathbb N$ successive samples of $y_d$, and then performing a least-squares algorithm, such as the one proposed in~\cite{ansari2019deadbeat}, to construct an estimate $\hat u_a$ of the input. Specifically, the attacker solves the following least-squares problem:\begin{equation}
    \hat u_a \doteq \arg\min_{\ul{\hat u}(k-1)}  \frac{1}{2} \left\| \ul{y_d}(k) - \Phi \ul{\hat u}(k-1) \right\|^2,
\end{equation}
where $\ul{y_d}(k) = \col({y}_d(k-r),\dots,{y}_d(k))$, $\ul{\bar y}(k) = \col(\hat{u}(k-r),\dots,\hat{u}(k))$, and
$\Phi$ is composed of the Markov parameters of the system with respect to the input matrix.
For $r$ larger than the relative degree of the system, and for $\D = 0$, this algorithm gives exact reconstruction of $u$ with some delay.
To illustrate the effects of injecting $Dd$ onto the communicated values of $y$, we consider two cases: a. the case in which $d = 0$; b. the case in which a disturbance signal is applied.

\subsection{No disturbance}
The results of the first scenario, in which no dithering signal is injected, are presented in Figure~\ref{fig:Wdeq0}. In the figure, a comparison of the unknown input with various estimates is provided, elementwise. We show the estimates $\hat u$, defined by the defender, and $\hat u_a$, defined by the attacker as in Definition~\ref{def:atk}. As expected, given the left-invertibility of the system, all estimates are accurate. 
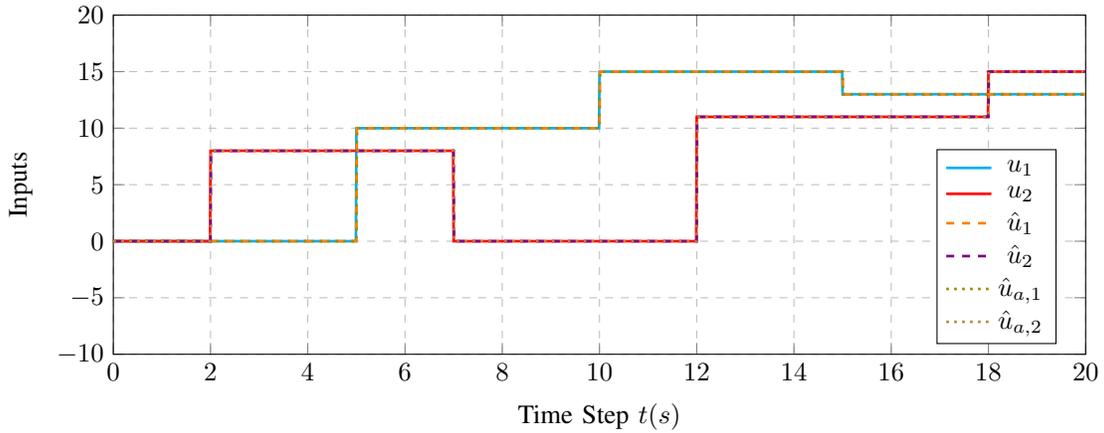
\begin{figure*}[t]
	\centering
	\begin{tikzpicture}
		\begin{axis}[
            xlabel={Time Step $t(s)$},
    		ylabel={Inputs},
			xmin=0, xmax=20,
			ymin=-10, ymax=20,
			xtick={0,2,4,6,8,10,12,14,16,18,20},
			ytick={-10,-5,0,5,10,15,20},
			legend pos=south east,
			ymajorgrids=true,
			grid=both,
			grid style=dashed,
			width=0.8*\textwidth,
			height=0.42*0.8*\textwidth
		]
                
			\addplot[ line width=1pt, solid, color=cyan ] table {fig/Wdeq0/mu1.txt};
			\addlegendentry{$u_1$}

            \addplot[ line width=1pt, solid, color=red ] table {fig/Wdeq0/mu2.txt};
			\addlegendentry{$u_2$}

            \addplot[ line width=1pt, dashed, color=orange] table {fig/Wdeq0/muHat1.txt};
			\addlegendentry{$\hat u_1$}

            \addplot[ line width=1pt, dashed, color=violet ] table {fig/Wdeq0/muHat2.txt};
			\addlegendentry{$\hat u_2$}

            \addplot[ line width=1pt, dotted, color=olive ] table {fig/Wdeq0/muHat_naive1.txt};
			\addlegendentry{$\hat u_{a,1}$}

            \addplot[ line width=1pt, dotted, color=brown ] table {fig/Wdeq0/muHat_naive2.txt};
			\addlegendentry{$\hat u_{a,2}$}
   
		\end{axis}
	\end{tikzpicture}
	\caption{Actual input ($u_i, i\in \{1,2\}$) and its reconstruction from the defender ($\hat u_i, i\in \{1,2\}$) and the least-square eavesdropper ($\hat u_{a,i}, i\in \{1,2\}$), without injected signal, i.e., $Dd(k) = 0\ \forall k$.}
    \label{fig:Wdeq0}
\end{figure*}

\subsection{Injected disturbance}
In this second scenario, we consider the case in which the dithering signal $Dd$ is injected into the output. Here we choose $D$ following the steps stated in Appendix~\ref{app:calcD}, which yields
\begin{equation*}
    \mathtt D = \begin{bmatrix}
        0.7957  &  0.5072 \\
   -0.0220 &  -0.0140 \\
   -0.0000 &  -0.0000 \\
   -0.5445 &  -0.4655 \\
    0.0135 &   0.0116 \\
    0.0000 &   0.0000
    \end{bmatrix}.
\end{equation*}
Recall that the results presented in this paper do not rely on any specific realization of $d(k)$, so long as a valid $\D$ (and therefore $\mathtt D$) is designed. 
To substantiate the results in the remainder of this section, however, let us propose the following definition:
\begin{equation}\label{eq:simu:d}
    \begin{split}
        &\check{x}^+ = \mathtt A\check{x} + \mathtt B \check{u}\\
        &d = \mathtt{D}^\dagger (\mathtt C\check x - y)
    \end{split}
\end{equation}
where $\check x$ is a \textit{simulated} state that has the same dynamics as the system, $\check x(0) = 0$, and $\check u$ is a fictitious input meant to obfuscate the true value of $u$. The design of this injected signal is inspired by \textit{covert attacks}, a class of data injection attacks that can disrupt certain anomaly detection algorithms~\cite{teixeira2015secure}.
As mentioned, the definition \eqref{eq:simu:d} is provided as a practical example, and further investigation of how to \textit{best} design $d(k)$ for obfuscation purposes is left for future work.

As shown in Figure~\ref{fig:Wdneq0}, the eavesdropper using the least square technique is no longer capable of estimating the unknown input without a bias, but the defender is still capable of reconstructing the input $u$ precisely. 
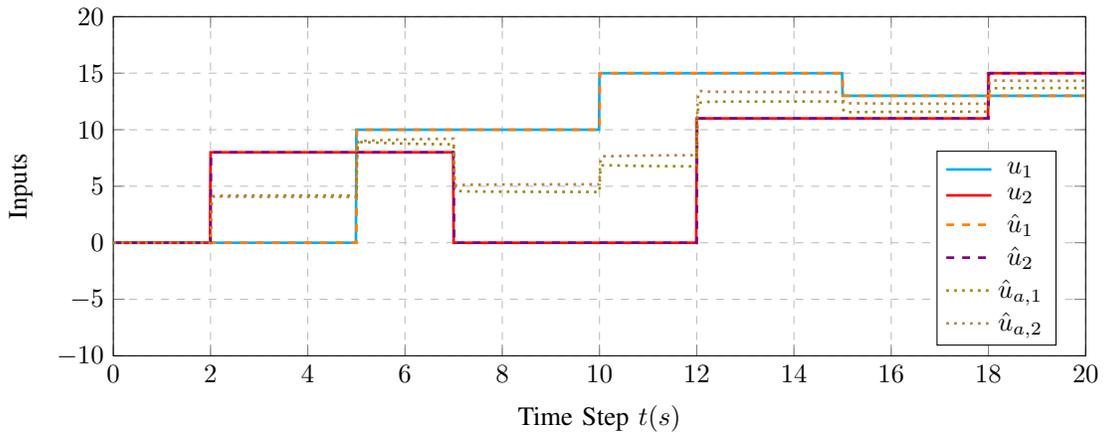
\begin{figure*}[t]
	\centering
	\begin{tikzpicture}
		\begin{axis}[
            xlabel={Time Step $t(s)$},
    		ylabel={Inputs},
			xmin=0, xmax=20,
			ymin=-10, ymax=20,
			xtick={0,2,4,6,8,10,12,14,16,18,20},
			ytick={-10,-5,0,5,10,15,20},
			legend pos=south east,
			ymajorgrids=true,
			grid=both,
			grid style=dashed,
			width=0.8*\textwidth,
			height=0.42*0.8*\textwidth
		]
                
			\addplot[ line width=1pt, solid, color=cyan ] table {fig/Wdneq0/mu1.txt};
			\addlegendentry{$u_1$}

            \addplot[ line width=1pt, solid, color=red ] table {fig/Wdneq0/mu2.txt};
			\addlegendentry{$u_2$}

            \addplot[ line width=1pt, dashed, color=orange] table {fig/Wdneq0/muHat1.txt};
			\addlegendentry{$\hat u_1$}

            \addplot[ line width=1pt, dashed, color=violet ] table {fig/Wdneq0/muHat2.txt};
			\addlegendentry{$\hat u_2$}

            \addplot[ line width=1pt, dotted, color=olive ] table {fig/Wdneq0/muHat_naive1.txt};
			\addlegendentry{$\hat u_{a,1}$}

            \addplot[ line width=1pt, dotted, color=brown ] table {fig/Wdneq0/muHat_naive2.txt};
			\addlegendentry{$\hat u_{a,2}$}

		\end{axis}
	\end{tikzpicture}
	\caption{Actual input ($u_i, i\in \{1,2\}$) and its reconstruction from the defender ($\hat u_i, i\in \{1,2\}$) and the least-square eavesdropper ($\hat u_{a,i}, i\in \{1,2\}$), with injected signal, i.e., $Dd(k) \neq 0\ \forall k$.}
    \label{fig:Wdneq0}
\end{figure*}
%
%

\section{Conclusions and Future Work}
\label{sec:end}
In this paper, we have introduced the concept of output redundancy in LTI systems and investigated the resulting properties from a geometric approach perspective. We have differentiated between Strong Output Redundancy (SOR) and Weak Output Redundancy (WOR) based on the dimensions of the image of the output map and that of the output space. For the SOR case, we have shown that the output redundancy resides in the quotient space of the output space modulo the image of the output map. Conversely, in the WOR scenario, output redundancy is associated with the infimal unobservability subspace that contains the image of the input map.
Throughout the paper, we have focused on a synthesis problem involving the concealment of  input  information by means of output allocation. We have shown that output redundancy in indeed necessary, and provided constructive solution for the considered classes of output-redundant systems. For future work, we aim to analyze the effectiveness of the proposed output allocation strategy against attackers employing more advanced reconstruction strategies. Additionally, exploring how output redundancy can be leveraged to enhance output-feedback controller design offers another stimulating direction.

\bibliographystyle{IEEEtran}
\bibliography{refs}

\appendix
\subsection{Numerical algorithm for calculating the matrix~$\mathtt L$}\label{app:calcL}
In this section, we introduce the algorithm for calculating $\mathtt L$, the matrix representation of $L \in \textbf{L}(\Ss^*)$ regarding System~\eqref{eq:system} while enabling the spectrum of $\bar A_L$ and $A_L^\flat$ ($A_L\coloneqq A+LC$) to be freely assigned. More specifically, we intend to find an $\mathtt L \in \R^{n \times p}$ such that
\begin{equation}\label{eq:spec_L}
     P A_L \Ss^* = 0, \quad
    \sigma \left( \bar A_L \right) = \bar \Lambda, \quad
    \sigma ( A_L^\flat ) = \Lambda^\flat,
\end{equation}
where $\bar \Lambda$ and $\Lambda^\flat$ denote symmetric sets of arbitrary complex numbers with compatible set numbers, and $P:\X \rightarrow \bar \X$ is the canonical projection as shown in Fig.~\ref{fig:weak}.

We first calculate a basis of $\Ss^*$ ($\dim (\Ss^*) = s$) using the unobservability subspace algorithm \cite[Theorem~18]{massoumnia1986geometric}, which is a matrix representation $\mathtt S \in \R^{n \times s}$ of the insertion map $S:\Ss^* \to \X$. Then, it is easy to calculate $\mathtt P \in \R^{(n-s) \times n}$ based on $\mathtt S$: 
\begin{equation*}
    \mathtt P \mathtt S = \mathtt 0_{(n-s) \times s}, \quad \rank \mathtt P = n-s.   
\end{equation*}
Subsequently, we find an arbitrary $\mathtt{L_0} \in \R^{n \times p}$ satisfying 
\begin{equation*}
    \mathtt P  (\mathtt A + \mathtt{L_0} \mathtt C)\mathtt S  = \mathtt 0_{(n-s)\times s }
\end{equation*}
We thereby compute $\mathtt{A_0}$ as
\begin{equation*}
    \mathtt{A_0} = \mathtt P (\mathtt A + \mathtt{L_0} \mathtt C) \mathtt P^\dagger.
\end{equation*}
Similar to calculating $\mathtt P$, we compute $\mathtt{P_Y} \in \R^{(p-z) \times p}$ ($\dim (C\Ss^*) = z$) by
\begin{equation*}
    \mathtt{P_Y} \mathtt{S_Y} = \mathtt 0_{(p-z) \times z}, \ \rank \mathtt{P_Y} = p-z,   
\end{equation*}
where $\mathtt{S_Y}\in \R^{p \times z}$ is a basis of the subspace $C\Ss^*$.
Moreover, calculate $\mathtt{C_0}:= \mathtt{P_Y} \mathtt C \mathtt P^\dagger$, and choose $\mathtt{L_{0a}}$ such that $\sigma(\mathtt{A_0} + \mathtt{L_{0a}C_0})= \bar \Lambda$.
After that, we let $\mathtt{L_1} = \mathtt{L_0} + \mathtt P^\dagger \mathtt{L_{0a} P_Y}$, where it can be verified that 
\begin{equation*}
    \sigma(\mathtt P (\mathtt A + \mathtt{L_1} \mathtt C) \mathtt P^\dagger) = \bar \Lambda.   
\end{equation*}
Now we define
\begin{equation*}
\begin{aligned}
    \mathtt{A_1} &\coloneqq \mathtt S^\dagger (\mathtt A + \mathtt{L_1} \mathtt C) \mathtt S    \\
    \mathtt{C_1} &\coloneqq \mathtt C \mathtt S.
\end{aligned}
\end{equation*}
Clearly, $\mathtt{A_1}$ and $\mathtt{C_1}$ are the matrix representation of the map $(A+L_1C)|\Ss^*$ and $C|\Ss^*$, respectively. Since $(\mathtt{C_1}, \mathtt{A_1})$ is observable (Proposition~\ref{lem:both-obsv}), there exists an $\mathtt{L_{1a}} \in \R^{s\times p}$ such that 
\begin{equation*}
    \sigma (\mathtt{A_1} + \mathtt{L_{1a}} \mathtt{C_1}) = \Lambda^\flat.
\end{equation*}
Now we let $\mathtt L = \mathtt{L_1} + \mathtt{SL_{1a}}$, which completes the computation of $\mathtt L$ satisfying \eqref{eq:spec_L}.
\subsection{Numerical algorithm for calculating the matrix~$\mathtt D$}\label{app:calcD}
%
Suppose we have calculated $\mathtt L$ as described in Appendix~\ref{app:calcL} with the desired spectrum for $A_L^\flat$ and $\bar A_L$. 
Then, we calculate $\mathtt{\bar L^\flat} = \mathtt{PLS_Y} \in \R^{(n-s)\times z}$, and we obtain the rank of $\mathtt{\bar L^\flat}$ as $\tau$.
Then we select a matrix $\mathtt{D^\flat_a} \in \R^{z \times \tau}$ such that
\begin{equation*}
    \rank (\mathtt{\bar L^\flat} \mathtt{D^\flat_a} ) = \tau.
\end{equation*}
We calculate a basis of $\bar \V^*$ (defined in \eqref{eq:def:barV*}) based on \cite[Theorem~4.3]{wonham1985linear}, denoted by $\mathtt V \in \R^{(n-p)\times v}$, where
\begin{equation*}
    \mathtt{\bar A_L} = \mathtt{P (A+LC) P}^\dagger, \quad 
    \mathtt{\bar D_{L_a}} = \mathtt{\bar L^\flat} \mathtt{D^\flat_a}, \quad
    \mathtt{\bar C} = \mathtt {P_Y C P}^\dagger .
\end{equation*}
Subsequently, we check the solution of
\begin{equation}\label{eq:imVimLW}
    \mathtt{V} \mathtt{\theta}_i = \mathtt{\bar D_{L_a}}[i],\ i \in \{1,\ldots , \tau\}
\end{equation}
where $\mathtt{\bar D_{L_a}}[i]$ represents the $i$-th column of the matrix $\mathtt{\bar D_{L_a}}$. If \eqref{eq:imVimLW} cannot be solved for any $i \in \{1,\ldots , \tau\}$, then we calculate $\mathtt D = \mathtt{S_Y D^\flat_a}$, and we obtain the desired $\mathtt D$.

Otherwise, we record the index $i$ where \eqref{eq:imVimLW} has a solution for $\theta_i$, and we discard the corresponding columns of $\mathtt{\bar D_{L_a}}$, resulting in a new matrix $\mathtt{\bar D_{L_b}} \in \R^{(n-p) \times \tau_b}$ with fewer number of columns ($\tau_b < \tau$). Then, we compute $\mathtt D = \mathtt{S_Y D^\flat_a}$, and we obtain the desired $\mathtt D$.

\begin{IEEEbiography}[{\includegraphics[width=1in,height=1.25in,clip,keepaspectratio]{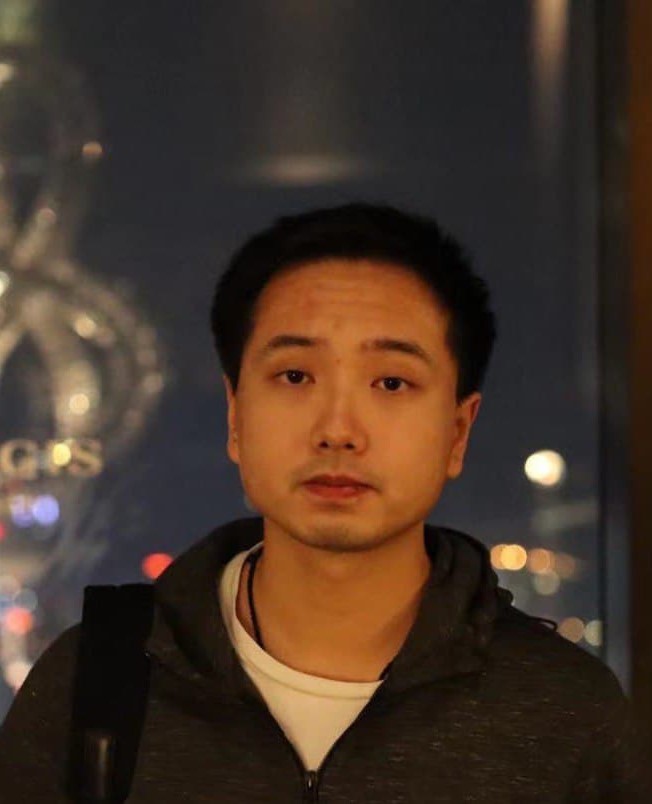}}]{Guitao Yang}
received the B.Eng. degree in Electrical and Electronic Engineering from The University of Manchester, Manchester, UK, in 2016. He received the M.Sc. and Ph.D. degrees in Control Systems from Imperial College London, London, U.K., in 2017 and 2022, respectively. He is currently a Research Associate with the Department of Electrical and Electronic Engineering, Imperial College London. His research interests include distributed state estimation, fault-tolerant observers, and the geometric approach.
\end{IEEEbiography}

\vspace{-0.5 true cm}

\begin{IEEEbiography}
[{\includegraphics[width=1in,height=1.25in,clip,keepaspectratio]{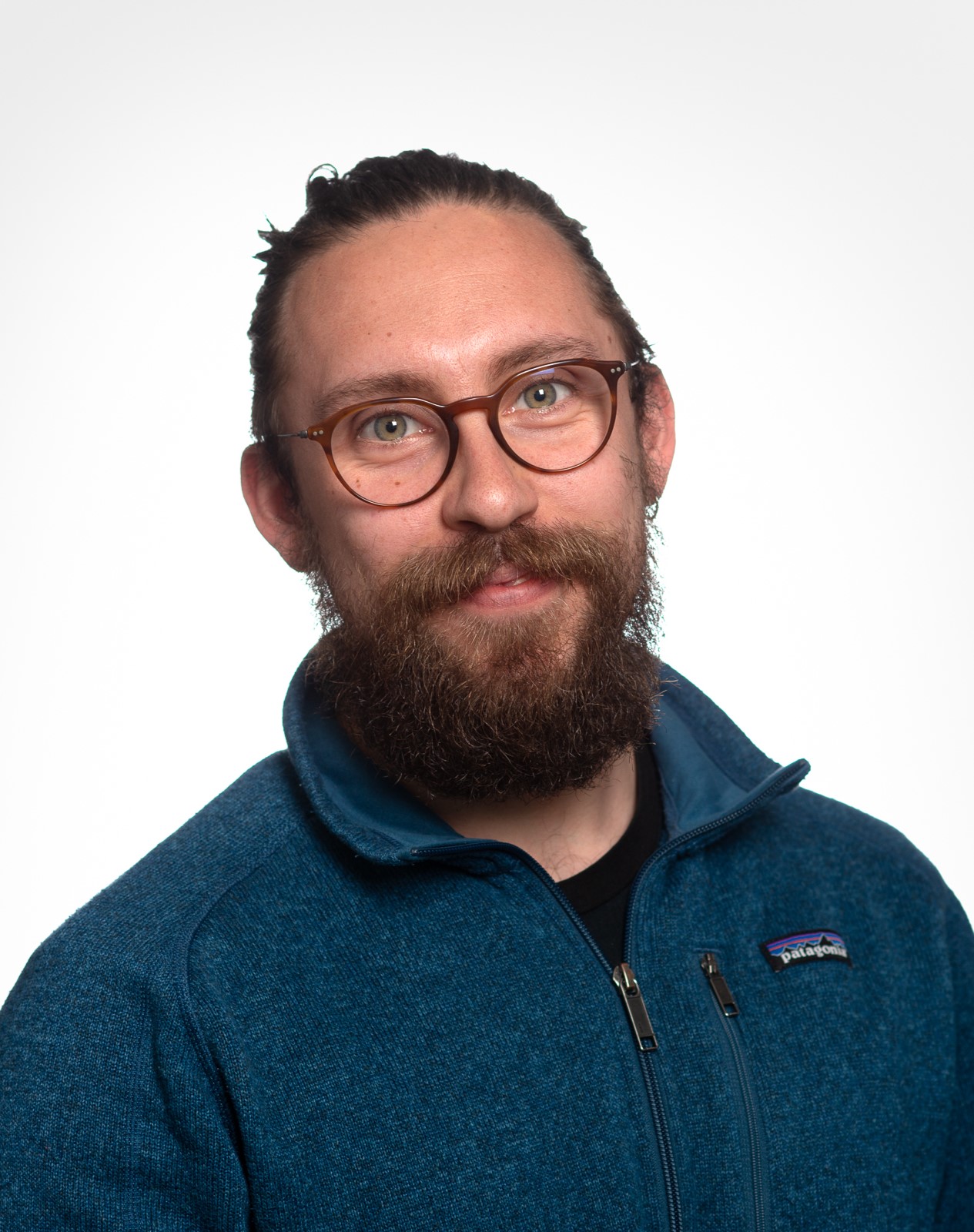}}]{Alexander J. Gallo}
received the M.Eng. in Electrical and Electronic Engineering from Imperial College London, London, UK, in 2016. He received his Ph.D. degree in Control Engineering from Imperial College London, London, UK, in 2021. From 2021 to 2024 he was a postdoctoral researcher at the Delft Center for Systems and Control, Delft University of Technology, Delft, the Netherlands. 
His main research interests include distributed cyber-security and fault tolerant control for large-scale interconnected systems, with a particular focus on energy distribution networks, as well as health-aware and fault tolerant control of wind turbines.
\end{IEEEbiography}

\vspace{-0.8 true cm}

\begin{IEEEbiography}[{\includegraphics[width=1in,height=1.25in,clip,keepaspectratio]{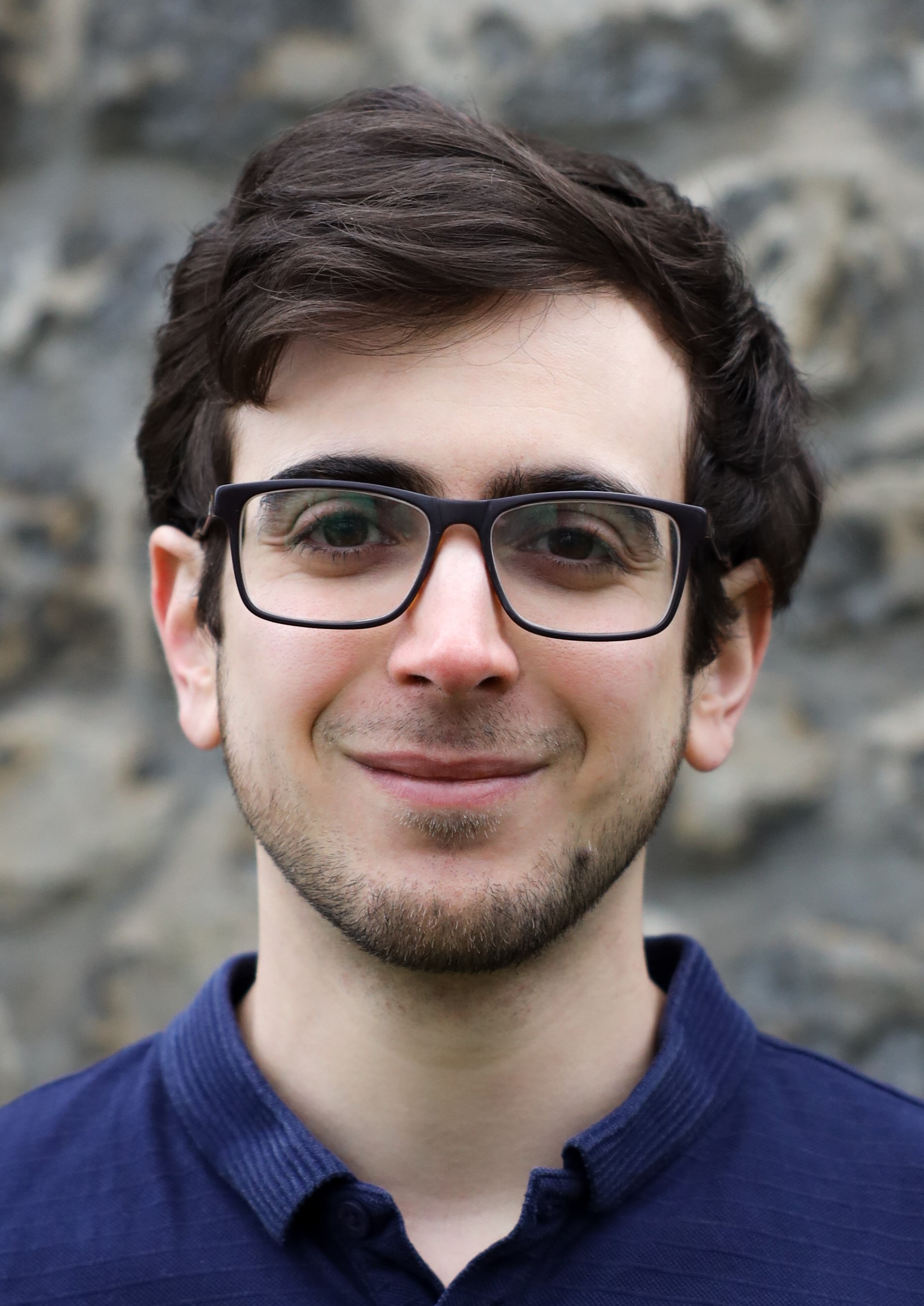}}]
{Angelo Barboni} received the M.Sc. in Electrical and Control Engineering, with honors, from the University of Trieste, Italy, in 2015, and the Ph.D. degree in Control Engineering from Imperial College London, UK, in 2021, where he was awarded a scholarship with the High-Performance Embedded and Distributed Systems (HiPEDS) Centre for Doctoral Training. 
From 2021 to 2023, he was a Research Associate with the Control and Power group in the Department of Electrical and Electronic Engineering, Imperial College London. 
Since 2023, he is employed as a researcher within Group Accumulation Management with the Zurich Insurance Group, Zurich, Switzerland, working on models for risk quantification in the cyber security domain. 
Besides security, his scientific interests include (distributed) estimation, anomaly detection, fault-tolerant control, and optimization.
\end{IEEEbiography}

\vspace{-0.8 true cm}

\begin{IEEEbiography}[{\includegraphics[width=1in,height=1.25in,clip,keepaspectratio]{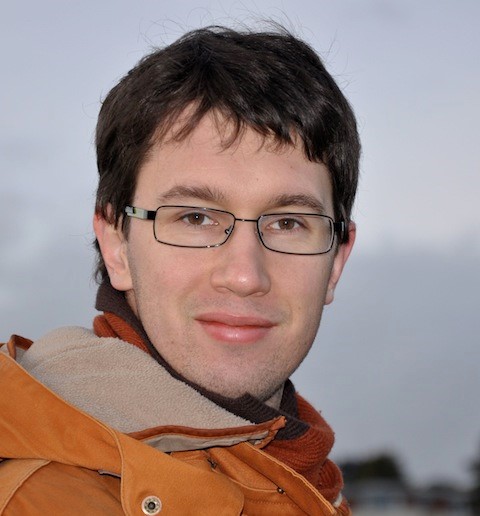}}]{Riccardo Ferrari} (Senior Member, IEEE)
received the Laurea degree with printing honors and the Ph.D. degree from University of Trieste, Italy. He held both academic and industrial R\&D positions, in particular as researcher in the field of process instrumentation and control for the steel-making sector. He is a Marie Curie alumnus and currently an Associate Professor with the Delft Center for Systems and Control, Delft University of Technology, The Netherlands. He is a co-recipient of the O. Hugo Schuck Award 2023. 
His research interests include fault tolerant control and fault diagnosis and attack detection in large-scale cyber–physical systems, with applications to wind energy generation, electric mobility, and cooperative autonomous vehicles.
\end{IEEEbiography}

\vspace{-0.5 true cm}

\begin{IEEEbiography}[{\includegraphics[width=1in,height=1.25in,clip,keepaspectratio]{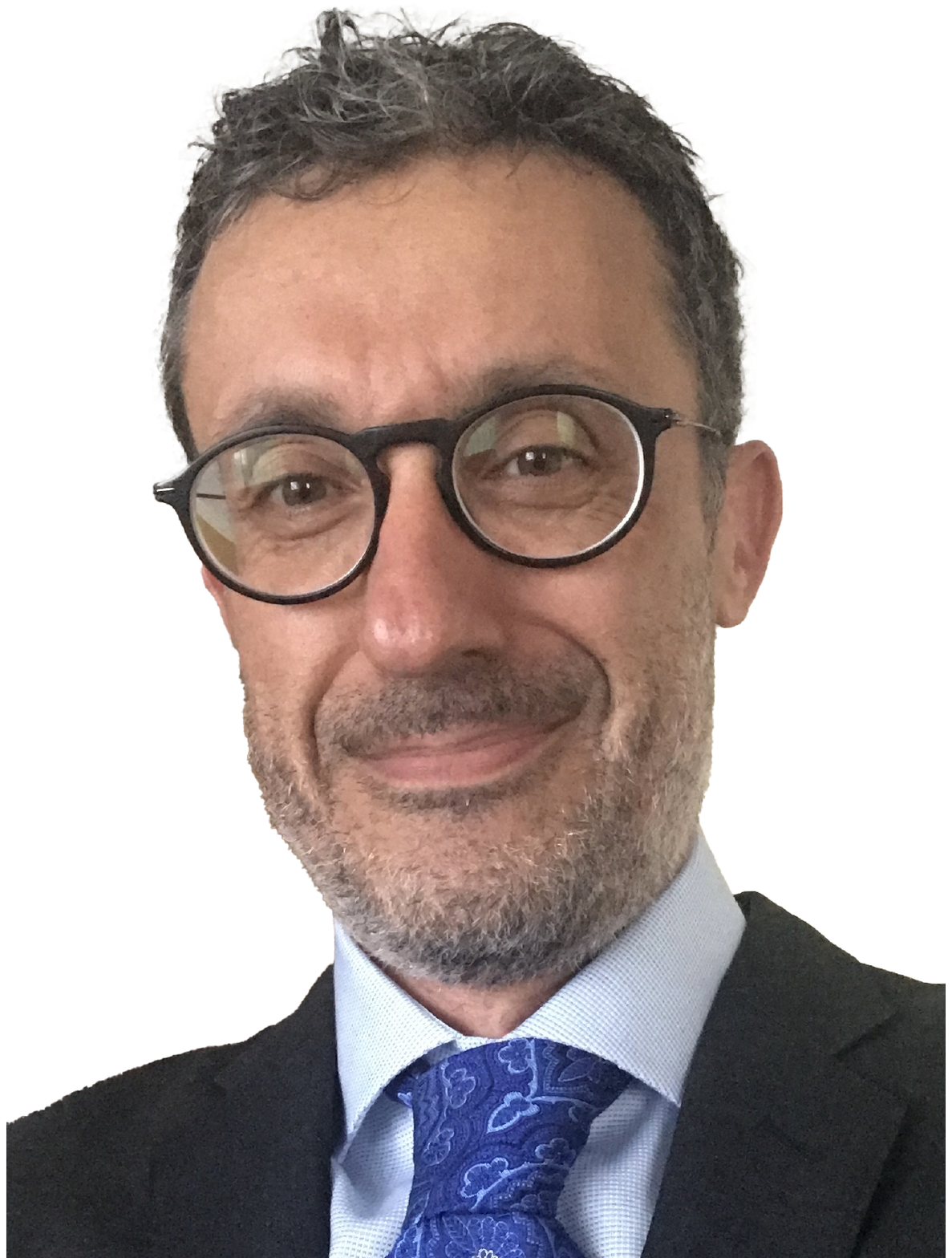}}]{Andrea Serrani} (Member, IEEE)
received the Laurea (B.Eng.) degree in Electrical Engineering in 1993, and the Ph.D. degree in 1997 from the University of Ancona, Italy. From 1994 to 1999, he was a Fulbright Fellow at Washington University in St. Louis, MO, where he obtained the M.S. and D.Sc. degrees in Systems Science and Mathematics in 1996 and 2000, respectively. Since 2002, he has been with the Department of Electrical and Computer Engineering of The Ohio State University, where he is currently a Professor and Associate Chair.
He has held visiting positions at the Universities of Bologna and Padua, Italy, and multiple summer faculty fellowships at AFRL. The research activity of Prof. Serrani lies at the intersection of nonlinear, adaptive and geometric control theory with applications in aerospace and marine systems, fluidic systems, robotics and automotive engineering. His work has been supported by AFRL, NSF, Ford Motor Co. and NASA, among others.
Prof. Serrani has authored or co-authored more than 150 articles in journals, proceedings of international conferences and book chapters, and is the co-author of the book {\em Robust Autonomous Guidance: An Internal Model Approach} published by Springer-Verlag. Prof. Serrani was a Distinguished Lecturer of the IEEE CSS, and served as Editor-in-Chief of the {\em IEEE Trans. on Control Systems Technology} (2017-2024) and as an Associate Editor for the same journal (2010-2016), {\em Automatica} (2008-2014) and the {\em Int. Journal of Robust and Nonlinear Control} (2006-2014). He served on the Conference Editorial Boards of IEEE CSS and IFAC, as Program Chair for the 2019 ACC, and as General Co-chair for the 2022 IEEE CDC. Currently, he serves as VP-Publications and on the Board of Governors of the IEEE CSS.
\end{IEEEbiography}

\vspace{-0.5 true cm}

\begin{IEEEbiography}[{\includegraphics[width=1in,height=1.25in,clip,keepaspectratio]{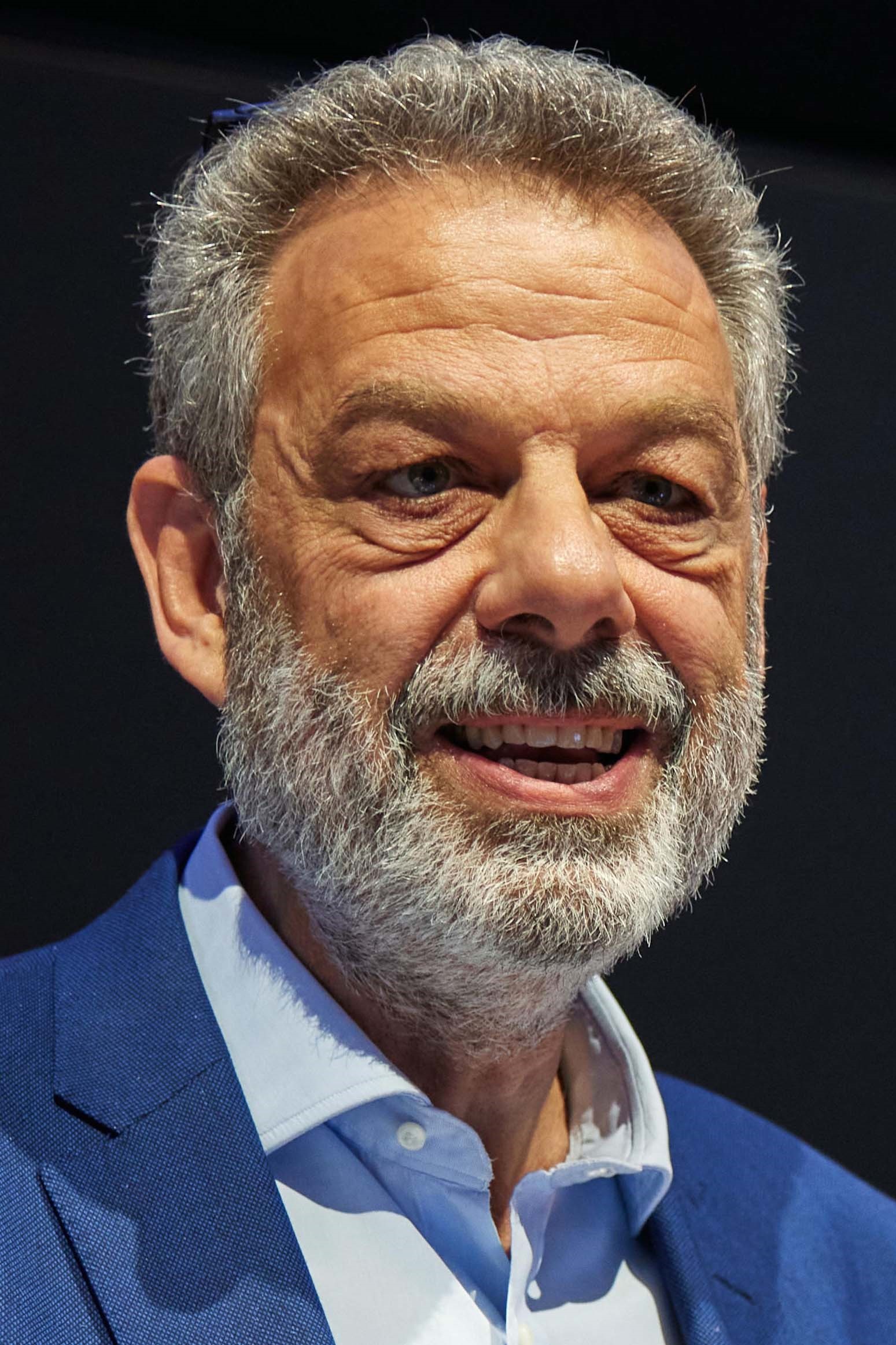}}]{Thomas Parisini} (Fellow, IEEE) received the Ph.D. degree in electronic engineering and computer science from the University of Genoa, Italy, in 1993. He was an Associate Professor at Politecnico di Milano, Milano, Italy. He currently holds the Chair of Industrial Control, and is the Head of the Control and Power Research Group, Imperial College London, London, U.K. He also holds a Distinguished Professorship at Aalborg University, Denmark. Since 2001, he has been the Danieli Endowed Chair of automation engineering with the University of Trieste, Trieste, Italy, where from 2009 to 2012, he was the Deputy Rector.  In 2023, he held a “Scholar-in-Residence” visiting position with Digital Futures-KTH, Stockholm, Sweden. He has authored or coauthored a research monograph in the Communication and Control Series, Springer Nature, and more than 400 research papers in archival journals, book chapters, and international conference proceedings. Dr. Parisini was the recipient of the Knighthood of the Order of Merit of the Italian Republic for scientific achievements abroad awarded by the Italian President of the Republic in 2023. In 2018 he received the Honorary Doctorate from the University of Aalborg, Denmark and in 2024, the IEEE CSS Transition to Practice Award. Moreover, we was awarded the 2007 IEEE Distinguished Member Award, and was co-recipient of the IFAC Best Application Paper Prize of the Journal of Process Control, Elsevier, for the three-year period 2011–2013 and of the 2004 Outstanding Paper Award of IEEE TRANSACTIONS ON NEURAL NETWORKS. In 2016, he was awarded as Principal Investigator with Imperial of the H2020 European Union flagship Teaming Project KIOS Research and Innovation Centre of Excellence led by the University of Cyprus with an overall budget of over 40 million Euros. He was the 2021–2022 President of the IEEE Control Systems Society and he was the Editor-in-Chief of IEEE TRANSACTIONS ON CONTROL SYSTEMS TECHNOLOGY (2009–2016). He was the Chair of the IEEE CSS Conference Editorial Board (2013–2019). He was the associate editor of several journals including the IEEE TRANSACTIONS ON AUTOMATIC CONTROL and the IEEE TRANSACTIONS ON NEURAL NETWORKS. He is currently an Editor of Automatica and the Editor-in-Chief of the European Journal of Control. He was the Program Chair of the 2008 IEEE Conference on Decision and Control and General Co-Chair of the 2013 IEEE Conference on Decision and Control. He is a Fellow of IFAC. He is a Member of IEEE TAB Periodicals Review and Advisory Committee.
\end{IEEEbiography}

\end{document}

%% file: cmds.tex
\def\R{\mathbb{R}}

\def\X{\mathscr{X}}
\def\Y{\mathscr{Y}}
\def\B{\mathscr{B}}
\def\D{\mathscr{D}}
\def\U{\mathscr{U}}

\def\V{\mathscr{V}}
\def\N{\mathscr{N}}

\def\Ss{\mathscr{S}}

\def\Rs{\mathscr{R}}

\def\W{\mathscr{W}}
\def\K{\mathscr{K}}

\newcommand{\ang}[1]{\left\langle #1\right\rangle}

\newcommand{\matrices}[1]{\begin{bmatrix} #1\end{bmatrix}}

\newcommand{\ul}[1]{\,\underline{#1}} 

\usepackage{amsthm}

\newtheoremstyle{mytheoremstyle} 
{\smallskipamount}                    
{\smallskipamount}                    
{\itshape}                   
{}                           
{\bfseries}                   
{.}                          
{.5em}                       
{}  %
\theoremstyle{mytheoremstyle}

\newtheorem{defn}{Definition}
\newtheorem{problem}{Problem}
\newtheorem{prop}{Proposition}
\newtheorem{assum}{Assumption}
\newtheorem{thm}{Theorem}
\newtheorem{lem}{Lemma}
\newtheorem{rem}{Remark}
\newtheorem{cor}{Corollary}

\DeclareMathOperator{\rank}{rank}
\DeclareMathOperator{\im}{Im}
\DeclareMathOperator{\ke}{Ker}

\DeclareMathOperator{\col}{col}

%% file: fig/sche.tikz
\begin{tikzpicture}[semithick]
    \tikzset{block/.style = {draw, inner xsep = 2ex}}

    \node[block] (plant) at (0,0) {$ \begin{aligned}
        x^+ &= Ax + Bu \\
        y &= C x\end{aligned}$};
    \node[block, 
        right=5ex of plant.east] (allocator) {$ y_d = y + D d$};
    \draw[-latex] (plant) -- (allocator) node[above, midway] {$y$};

    \coordinate[right=3ex of allocator.east] (ne);
    \coordinate (center_top) at ($(plant.west)!0.5!(allocator.east)$);    
    
    \node[block, 
        below=20ex of allocator.south east,
        anchor=east] (observer) {Inverse};
    \node[
        left=8ex of observer.west] (ux) {$\hat u$};
    \draw[-latex] (observer) -- (ux) node [above, midway] {};

    \coordinate[xshift = -6ex] (sw) at (plant.west |- ux.west);

    \node (network) at ($(center_top |- plant.south)!0.5!(observer.north)$) {Network broadcasting $y_d$};
    \def\nheight{3ex}
    \def\nxpad{4ex}
    \coordinate (ntop) at ($(network.north) + (0,\nheight)$);
    \coordinate (nbot) at ($(network.south) + (0,-\nheight)$);
    \draw[dashed] ($(ne|-ntop) + (\nxpad,0)$) -- ($(sw|-ntop) + (-\nxpad,0)$);
    \draw[dashed] ($(ne|-nbot) + (\nxpad,0)$) -- ($(sw|-nbot) + (-\nxpad,0)$);

    \draw[-latex] (allocator.east) -| (ne |- ntop) node [near end, right] {$y_d$};
    \draw[-latex] (ne |- nbot) |- (observer.east) node [near start, right] {$ y_d$} ;
    \node[
        left=5ex of plant.west] (u) {$u$};
    \draw[-latex] (u) -- (plant) node [above, midway] {};
\end{tikzpicture}

%% file: fig/strong_number.tikz
 
\tikzset{
pattern size/.store in=\mcSize, 
pattern size = 5pt,
pattern thickness/.store in=\mcThickness, 
pattern thickness = 0.3pt,
pattern radius/.store in=\mcRadius, 
pattern radius = 1pt}
\makeatletter
\pgfutil@ifundefined{pgf@pattern@name@_xyjy4qbyy}{
\pgfdeclarepatternformonly[\mcThickness,\mcSize]{_xyjy4qbyy}
{\pgfqpoint{0pt}{0pt}}
{\pgfpoint{\mcSize+\mcThickness}{\mcSize+\mcThickness}}
{\pgfpoint{\mcSize}{\mcSize}}
{
\pgfsetcolor{\tikz@pattern@color}
\pgfsetlinewidth{\mcThickness}
\pgfpathmoveto{\pgfqpoint{0pt}{0pt}}
\pgfpathlineto{\pgfpoint{\mcSize+\mcThickness}{\mcSize+\mcThickness}}
\pgfusepath{stroke}
}}
\makeatother

 
\tikzset{
pattern size/.store in=\mcSize, 
pattern size = 5pt,
pattern thickness/.store in=\mcThickness, 
pattern thickness = 0.3pt,
pattern radius/.store in=\mcRadius, 
pattern radius = 1pt}
\makeatletter
\pgfutil@ifundefined{pgf@pattern@name@_v0c7r9mdc}{
\pgfdeclarepatternformonly[\mcThickness,\mcSize]{_v0c7r9mdc}
{\pgfqpoint{0pt}{0pt}}
{\pgfpoint{\mcSize+\mcThickness}{\mcSize+\mcThickness}}
{\pgfpoint{\mcSize}{\mcSize}}
{
\pgfsetcolor{\tikz@pattern@color}
\pgfsetlinewidth{\mcThickness}
\pgfpathmoveto{\pgfqpoint{0pt}{0pt}}
\pgfpathlineto{\pgfpoint{\mcSize+\mcThickness}{\mcSize+\mcThickness}}
\pgfusepath{stroke}
}}
\makeatother

 
\tikzset{
pattern size/.store in=\mcSize, 
pattern size = 5pt,
pattern thickness/.store in=\mcThickness, 
pattern thickness = 0.3pt,
pattern radius/.store in=\mcRadius, 
pattern radius = 1pt}
\makeatletter
\pgfutil@ifundefined{pgf@pattern@name@_orz3y5q86}{
\pgfdeclarepatternformonly[\mcThickness,\mcSize]{_orz3y5q86}
{\pgfqpoint{0pt}{0pt}}
{\pgfpoint{\mcSize+\mcThickness}{\mcSize+\mcThickness}}
{\pgfpoint{\mcSize}{\mcSize}}
{
\pgfsetcolor{\tikz@pattern@color}
\pgfsetlinewidth{\mcThickness}
\pgfpathmoveto{\pgfqpoint{0pt}{0pt}}
\pgfpathlineto{\pgfpoint{\mcSize+\mcThickness}{\mcSize+\mcThickness}}
\pgfusepath{stroke}
}}
\makeatother
\tikzset{every picture/.style={line width=0.75pt}} 

\begin{tikzpicture}[x=0.75pt,y=0.75pt,yscale=-1,xscale=1]

\draw  (64.6,192.61) -- (343.8,192.61)(76.32,30.7) -- (76.32,210.6) (336.8,187.61) -- (343.8,192.61) -- (336.8,197.61) (71.32,37.7) -- (76.32,30.7) -- (81.32,37.7)  ;
\draw  [pattern=_xyjy4qbyy,pattern size=6pt,pattern thickness=0.75pt,pattern radius=0pt, pattern color={rgb, 255:red, 0; green, 0; blue, 0}] (106.9,140.25) -- (119.9,140.25) -- (119.9,192.75) -- (106.9,192.75) -- cycle ;
\draw  [pattern=_v0c7r9mdc,pattern size=6pt,pattern thickness=0.75pt,pattern radius=0pt, pattern color={rgb, 255:red, 0; green, 0; blue, 0}] (209.9,140.25) -- (222.9,140.25) -- (222.9,192.75) -- (209.9,192.75) -- cycle ;
\draw  [pattern=_orz3y5q86,pattern size=6pt,pattern thickness=0.75pt,pattern radius=0pt, pattern color={rgb, 255:red, 0; green, 0; blue, 0}] (306.4,80.75) -- (319.4,80.75) -- (319.4,192.75) -- (306.4,192.75) -- cycle ;
\draw  [dash pattern={on 0.84pt off 2.51pt}]  (77,140.2) -- (209.9,140.25) ;
\draw  [dash pattern={on 0.84pt off 2.51pt}]  (77,80.6) -- (306.4,80.75) ;

\draw (78.32,196.01) node [anchor=north west][inner sep=0.75pt]    {$ \begin{array}{l}
{\textstyle \dim(\mathrm{Im} \, B)}\\
={\textstyle \dim(\mathscr{U})}
\end{array}$};
\draw (175.34,201.34) node [anchor=north west][inner sep=0.75pt]    {\large ${\textstyle \dim(\mathrm{Im} \, C)}$};
\draw (277.34,200.84) node [anchor=north west][inner sep=0.75pt]    {\large ${\textstyle \dim(\mathscr{Y})}$};
\draw (56,136.8) node [anchor=north west][inner sep=0.75pt]    {\large $m$};
\draw (61,76.8) node [anchor=north west][inner sep=0.75pt]    {\large $p$};

\end{tikzpicture}

%% file: fig/weak_number.tikz
 
\tikzset{
pattern size/.store in=\mcSize, 
pattern size = 5pt,
pattern thickness/.store in=\mcThickness, 
pattern thickness = 0.3pt,
pattern radius/.store in=\mcRadius, 
pattern radius = 1pt}
\makeatletter
\pgfutil@ifundefined{pgf@pattern@name@_b7x6a1bhh}{
\pgfdeclarepatternformonly[\mcThickness,\mcSize]{_b7x6a1bhh}
{\pgfqpoint{0pt}{0pt}}
{\pgfpoint{\mcSize+\mcThickness}{\mcSize+\mcThickness}}
{\pgfpoint{\mcSize}{\mcSize}}
{
\pgfsetcolor{\tikz@pattern@color}
\pgfsetlinewidth{\mcThickness}
\pgfpathmoveto{\pgfqpoint{0pt}{0pt}}
\pgfpathlineto{\pgfpoint{\mcSize+\mcThickness}{\mcSize+\mcThickness}}
\pgfusepath{stroke}
}}
\makeatother

 
\tikzset{
pattern size/.store in=\mcSize, 
pattern size = 5pt,
pattern thickness/.store in=\mcThickness, 
pattern thickness = 0.3pt,
pattern radius/.store in=\mcRadius, 
pattern radius = 1pt}
\makeatletter
\pgfutil@ifundefined{pgf@pattern@name@_vllqvu7ih}{
\pgfdeclarepatternformonly[\mcThickness,\mcSize]{_vllqvu7ih}
{\pgfqpoint{0pt}{0pt}}
{\pgfpoint{\mcSize+\mcThickness}{\mcSize+\mcThickness}}
{\pgfpoint{\mcSize}{\mcSize}}
{
\pgfsetcolor{\tikz@pattern@color}
\pgfsetlinewidth{\mcThickness}
\pgfpathmoveto{\pgfqpoint{0pt}{0pt}}
\pgfpathlineto{\pgfpoint{\mcSize+\mcThickness}{\mcSize+\mcThickness}}
\pgfusepath{stroke}
}}
\makeatother

 
\tikzset{
pattern size/.store in=\mcSize, 
pattern size = 5pt,
pattern thickness/.store in=\mcThickness, 
pattern thickness = 0.3pt,
pattern radius/.store in=\mcRadius, 
pattern radius = 1pt}
\makeatletter
\pgfutil@ifundefined{pgf@pattern@name@_7m0zh1jjw}{
\pgfdeclarepatternformonly[\mcThickness,\mcSize]{_7m0zh1jjw}
{\pgfqpoint{0pt}{0pt}}
{\pgfpoint{\mcSize+\mcThickness}{\mcSize+\mcThickness}}
{\pgfpoint{\mcSize}{\mcSize}}
{
\pgfsetcolor{\tikz@pattern@color}
\pgfsetlinewidth{\mcThickness}
\pgfpathmoveto{\pgfqpoint{0pt}{0pt}}
\pgfpathlineto{\pgfpoint{\mcSize+\mcThickness}{\mcSize+\mcThickness}}
\pgfusepath{stroke}
}}
\makeatother
\tikzset{every picture/.style={line width=0.75pt}} 

\begin{tikzpicture}[x=0.75pt,y=0.75pt,yscale=-1,xscale=1]

\draw  (64.6,192.61) -- (343.8,192.61)(76.32,30.7) -- (76.32,210.6) (336.8,187.61) -- (343.8,192.61) -- (336.8,197.61) (71.32,37.7) -- (76.32,30.7) -- (81.32,37.7)  ;
\draw  [pattern=_b7x6a1bhh,pattern size=6pt,pattern thickness=0.75pt,pattern radius=0pt, pattern color={rgb, 255:red, 0; green, 0; blue, 0}] (106.9,140.25) -- (119.9,140.25) -- (119.9,192.75) -- (106.9,192.75) -- cycle ;
\draw  [pattern=_vllqvu7ih,pattern size=6pt,pattern thickness=0.75pt,pattern radius=0pt, pattern color={rgb, 255:red, 0; green, 0; blue, 0}] (209.9,80.75) -- (222.9,80.75) -- (222.9,192.75) -- (209.9,192.75) -- cycle ;
\draw  [pattern=_7m0zh1jjw,pattern size=6pt,pattern thickness=0.75pt,pattern radius=0pt, pattern color={rgb, 255:red, 0; green, 0; blue, 0}] (306.4,80.75) -- (319.4,80.75) -- (319.4,192.75) -- (306.4,192.75) -- cycle ;
\draw  [dash pattern={on 0.84pt off 2.51pt}]  (77,140.2) -- (106.9,140.25) ;
\draw  [dash pattern={on 0.84pt off 2.51pt}]  (77,80.6) -- (306.4,80.75) ;

\draw (78.32,196.01) node [anchor=north west][inner sep=0.75pt]    {$ \begin{array}{l}
{\textstyle \dim(\mathrm{Im} \, B)}\\
={\textstyle \dim(\mathscr{U})}
\end{array}$};
\draw (175.34,201.34) node [anchor=north west][inner sep=0.75pt]    {\large ${\textstyle \dim(\mathrm{Im} \, C)}$};
\draw (277.34,200.84) node [anchor=north west][inner sep=0.75pt]    {\large ${\textstyle \dim(\mathscr{Y})}$};
\draw (56,136.8) node [anchor=north west][inner sep=0.75pt]    {\large $m$};
\draw (61,76.8) node [anchor=north west][inner sep=0.75pt]    {\large $p$};

\end{tikzpicture}

%% file: fig/strong.tikz
\begin{tikzpicture}


\def\xgap{2}
\def\ygap{1.5}
\foreach \i in {0,...,3}{
    \foreach \j in {0,...,3}{
        \coordinate (\j\i) at (\i*\xgap, \j*\ygap);
    }
}
\pgfmathatantwo{\ygap}{\xgap}
\def\labangle{\pgfmathresult}

\tikzset{label/.style = {font=\small, midway}}


\node (U) at (10) {$\mathscr U$};
\node (X) at (11) {$\mathscr X$};
\node (C) at (23) {$\im C$};
\node (X2) at (12) {$\mathscr X$};
\node (Y) at (13) {$\mathscr Y$};
\node (YmodC) at (03) {$\mathscr Y/\im C$};

\draw[-latex] (U) -- (X) node [label, above] {$B$};


\draw[-latex] (C) -- (Y) node [label, right] {$i$};
\draw[-latex] (X2) -- (C) node [label, above, rotate=\labangle] {$\im C|C$};

\draw[-latex] (X) -- (X2) node [label, above] {$A$};
\draw[-latex] (X2) -- (Y) node [label, above] {$C$};

\draw[-latex, dashed] (X2) -- (YmodC) node [label, below, rotate=-\labangle] {$0$};
\draw[-latex] (Y) -- (YmodC) node [label, right] {$\pi$};

\end{tikzpicture}

%% file: fig/thm-SOR.tikz
\begin{tikzpicture}


\def\xgap{2}
\def\ygap{2.25}
\foreach \i in {0,...,4}{
    \foreach \j in {0,...,3}{
        \coordinate (\j\i) at (\i*\xgap, \j*\ygap);
    }
}
\pgfmathatantwo{\ygap}{\xgap}
\def\labangle{\pgfmathresult}

\tikzset{label/.style = {font=\small, midway}}


\node (U) at (00) {$\mathscr U$};
\node (X) at (01) {$\mathscr X$};
\node (X2) at (02) {$\mathscr X$};
\node (Y) at (03) {$\mathscr Y$};
\node (Y2) at (04) {$\mathscr Y$};
\node (ImC) at (14) {$\im C$};

\draw[-latex] (U) -- (X) node [label, above] {$B$};
\draw[-latex] (X) -- (X2) node [label, above] {$A$};
\draw[-latex] (X2) -- (Y) node [label, above] {$C$};
\draw[-latex] (Y) -- (Y2) node [label, above] {$Q$};

\draw[-latex] (X2) -- (ImC) node [label, above, rotate=0.65*\labangle] {$\tilde C$};
\draw[-latex] (Y) -- (ImC) node [label, below, rotate=1.2*\labangle] {$\tilde Q$};
\draw[-latex] (ImC) -- (Y2) node [label, left] {$i$};






\end{tikzpicture}

%% file: fig/weak.tikz
\begin{tikzpicture}


\def\xgap{2}
\def\ygap{1.5}
\foreach \i in {0,...,3}{
    \foreach \j in {0,...,3}{
        \coordinate (\j\i) at (\i*\xgap, \j*\ygap);
    }
}
\pgfmathatantwo{\ygap}{\xgap}
\def\labangle{\pgfmathresult}

\tikzset{label/.style = {font=\small, midway}}


\node (U) at (10) {$\mathscr U$};
\node (XmodS) at (01) {$\bar{\mathscr X}$};
\node (X) at (11) {$\mathscr X$};
\node (S) at (21) {$\mathscr S^*$};
\node (S2) at (22) {$\mathscr S^*$};
\node (Z) at (23) {$C\mathscr S^*$};
\node (X2) at (12) {$\mathscr X$};
\node (Y) at (13) {$\mathscr Y$};
\node (XmodS2) at (02) {$\bar{\mathscr X}$};
\node (YmodZ) at (03) {$\bar{\mathscr Y}$};

\draw[-latex] (U) -- (S) node [label, above, rotate=\labangle] {$B^\flat$};
\draw[-latex] (S) -- (X) node [label, right] {$S$};
\draw[-latex] (U) -- (X) node [label, above] {$B$};
\draw[-latex] (X) -- (XmodS) node [label, right] {$P$};
\draw[-latex, dashed] (U) -- (XmodS) node [label, below, rotate=-\labangle] {$0$};

\draw[-latex] (S2) -- (Z) node [label, above] {$C^\flat$};

\draw[-latex] (S) -- (S2) node [label, above] {$A_L^\flat$};
\draw[-latex] (S2) -- (X2) node [label, right] {$S$};
\draw[-latex] (S2) -- (Y) node [label, above, rotate=-\labangle] {$C^\flat_Y$};
\draw[-latex] (Z) -- (Y) node [label, right] {$S_Y$};

\draw[-latex] (X) -- (X2) node [label, above] {$A_L$};
\draw[-latex] (X2) -- (Y) node [label, above] {$C$};

\draw[-latex] (X2) -- (XmodS2) node [label, right] {$P$};
\draw[-latex] (X2) -- (YmodZ) node [label, above, rotate=-\labangle] {$\bar C_Y$};
\draw[-latex] (Y) -- (YmodZ) node [label, right] {$P_Y$};

\draw[-latex] (XmodS) -- (XmodS2) node [label, above] {$\bar A_L$};
\draw[-latex] (XmodS2) -- (YmodZ) node [label, above] {$\bar{C}$};
\end{tikzpicture}

%% file: fig/injectD.tikz
\begin{tikzpicture}


\def\xgap{2}
\def\ygap{1.5}
\foreach \i in {0,...,3}{
    \foreach \j in {0,...,3}{
        \coordinate (\j\i) at (\i*\xgap, \j*\ygap);
    }
}
\pgfmathatantwo{\ygap}{\xgap}
\def\labangle{\pgfmathresult}

\tikzset{label/.style = {font=\small, midway}}


\node (D) at (10) {$\mathscr D$};
\node (YmodZ) at (01) {$\bar{\mathscr Y}$};
\node (Z) at (21) {$C\mathscr S^*$};
\node (Y) at (11) {$\mathscr Y$};

\draw[-latex] (D) -- (Z) node [label, above, rotate=\labangle] {$D^\flat$};
\draw[-latex] (D) -- (Y) node [label, above] {$D$};
\draw[-latex, dashed] (D) -- (YmodZ) node [label, below, rotate=-\labangle] {$0$};


\draw[-latex] (Z) -- (Y) node [label, right] {$S_Y$};

\draw[-latex] (Y) -- (YmodZ) node [label, right] {$P_Y$};

\end{tikzpicture}

%% file: fig/decomposeD.tikz
\begin{tikzpicture}


\def\xgap{2}
\def\ygap{1.5}
\foreach \i in {0,...,3}{
    \foreach \j in {0,...,3}{
        \coordinate (\j\i) at (\i*\xgap, \j*\ygap);
    }
}
\pgfmathatantwo{\ygap}{\xgap}
\def\labangle{\pgfmathresult}

\tikzset{label/.style = {font=\small, midway}}


\node (D) at (10) {$\mathscr D$};
\node (XmodS) at (01) {$\bar{\mathscr X}$};
\node (X) at (11) {$\mathscr X$};
\node (X2) at (12) {$\mathscr X$};
\node (Y) at (13) {$\mathscr Y$};
\node (XmodS2) at (02) {$\bar{\mathscr X}$};
\node (YmodZ) at (03) {$\bar{\mathscr Y}$};

\draw[-latex] (D) -- (X) node [label, above] {$D_L$};
\draw[-latex] (X) -- (XmodS) node [label, right] {$P$};
\draw[-latex] (D) -- (XmodS) node [label, below, rotate=-\labangle] {$\bar D_L$};



\draw[-latex] (X) -- (X2) node [label, above] {$A_L$};
\draw[-latex] (X2) -- (Y) node [label, above] {$C$};

\draw[-latex] (X2) -- (XmodS2) node [label, right] {$P$};
\draw[-latex] (Y) -- (YmodZ) node [label, right] {$P_Y$};

\draw[-latex] (XmodS) -- (XmodS2) node [label, above] {$\bar A_L$};
\draw[-latex] (XmodS2) -- (YmodZ) node [label, above] {$\bar C$};
\end{tikzpicture}

%% file: fig/SDbarX.tikz
\begin{tikzpicture}


\def\xgap{2}
\def\ygap{1.5}
\foreach \i in {0,...,3}{
    \foreach \j in {0,...,3}{
        \coordinate (\j\i) at (\i*\xgap, \j*\ygap);
    }
}
\pgfmathatantwo{\ygap}{\xgap}
\def\labangle{\pgfmathresult}

\tikzset{label/.style = {font=\small, midway}}


\node (X1) at (10) {$\mathscr X$};
\node (D) at (31) {$\mathscr D$};
\node (S) at (20) {$\mathscr S^*$};
\node (Z) at (21) {$C\mathscr{S}^*$};
\node (Y) at (11) {$\mathscr Y$};
\node (X2) at (12) {$\mathscr X$};
\node (XmodS) at (13) {$\bar{\mathscr X}$};

\draw[-latex] (X1) -- (Y) node [label, below] {$C$};
\draw[-latex] (S) -- (Z) node [label, above] {$C^\flat$};

\draw[-latex] (Y) -- (X2) node [label, below] {$L$};

\draw[-latex] (D) -- (Z) node [label, left] {$D^\flat$};
\draw[-latex] (Z) -- (Y) node [label, left] {$S_Y$};


\draw[-latex] (X2) -- (XmodS) node [label, below] {$P$};
\draw[-latex] (S) -- (X1) node [label, left] {$S$};

\draw[-latex] (Z) -- (XmodS) node [label, above, near start, rotate=-.6*\labangle] {$\bar{L}^\flat$}; 
\draw[-latex] (D) -- (XmodS) node [label, above, rotate=-1*\labangle] {$\bar{D}_L$}; 
\draw[-latex] (Z) -- (X2) node [label, below, rotate=-1*\labangle] {$L^\flat$}; 

\end{tikzpicture}

%% file: fig/D1.tikz
\begin{tikzpicture}


\def\xgap{3}
\def\ygap{2}
\foreach \i in {0,...,3}{
    \foreach \j in {0,...,3}{
        \coordinate (\j\i) at (\i*\xgap, \j*\ygap);
    }
}
\pgfmathatantwo{\ygap}{\xgap}
\def\labangle{\pgfmathresult}

\tikzset{label/.style = {font=\small, midway}}


\node (D) at (10) {$\mathscr D$};
\node (D1) at (20) {$\mathscr D_1$};
\node (CS) at (11) {$C\mathscr{S}^*$};
\node (barX) at (21) {$\bar \X$};

\draw[-latex] (D1) -- (D) node [label, left] {$i$};

\draw[-latex] (D) -- (CS) node [label, below] {$D^\flat$};

\draw[-latex] (CS) -- (barX) node [label, right] {$\bar L^\flat$};

\draw[-latex, densely dotted] (D1) -- (barX) node [label, above] {$\bar D^1_L$};

\draw[-latex, densely dotted] (D1) -- (CS) node [label, below, near end, rotate=-1*\labangle] {$D_1^\flat$}; 

\draw[-latex] (D) -- (barX) node [label, above, near end, rotate=1.2*\labangle] {$\bar D_L$}; 

\end{tikzpicture}

%% file: fig/DGU.tikz
\tikzset{every picture/.style={line width=0.75pt}} 

\begin{tikzpicture}[x=0.75pt,y=0.75pt,yscale=-1,xscale=1]

\draw   (80,269.09) -- (80,247.04) (67.61,242.14) -- (92.39,242.14) (80,242.14) -- (80,220.09) (73.8,249) -- (73.8,247.04) -- (86.2,247.04) -- (86.2,249) -- (73.8,249) -- cycle ;
\draw   (214,180) -- (224.58,180) -- (226.93,174.5) -- (231.63,185.5) -- (236.33,174.5) -- (241.03,185.5) -- (245.73,174.5) -- (250.43,185.5) -- (255.13,174.5) -- (259.83,185.5) -- (262.18,180) -- (272.75,180) ;
\draw   (309.75,180) -- (318.65,180) .. controls (318.78,177.15) and (320.02,174.71) .. (321.77,173.86) .. controls (323.51,173.01) and (325.42,173.91) .. (326.57,176.14) .. controls (327.45,177.88) and (327.81,180.13) .. (327.56,182.31) .. controls (327.56,183.17) and (327.11,183.86) .. (326.57,183.86) .. controls (326.02,183.86) and (325.58,183.17) .. (325.58,182.31) .. controls (325.32,180.13) and (325.68,177.88) .. (326.57,176.14) .. controls (327.59,174.29) and (329.03,173.24) .. (330.52,173.24) .. controls (332.02,173.24) and (333.45,174.29) .. (334.48,176.14) .. controls (335.36,177.88) and (335.72,180.13) .. (335.47,182.31) .. controls (335.47,183.17) and (335.03,183.86) .. (334.48,183.86) .. controls (333.93,183.86) and (333.49,183.17) .. (333.49,182.31) .. controls (333.24,180.13) and (333.6,177.88) .. (334.48,176.14) .. controls (335.51,174.29) and (336.94,173.24) .. (338.44,173.24) .. controls (339.93,173.24) and (341.37,174.29) .. (342.39,176.14) .. controls (343.28,177.88) and (343.64,180.13) .. (343.38,182.31) .. controls (343.38,183.17) and (342.94,183.86) .. (342.39,183.86) .. controls (341.85,183.86) and (341.4,183.17) .. (341.4,182.31) .. controls (341.15,180.13) and (341.51,177.88) .. (342.39,176.14) .. controls (343.54,173.91) and (345.45,173.01) .. (347.19,173.86) .. controls (348.94,174.71) and (350.18,177.15) .. (350.31,180) -- (359.21,180) ;
\draw    (272.75,180) -- (295.75,180) ;
\draw [shift={(298.75,180)}, rotate = 180] [fill={rgb, 255:red, 0; green, 0; blue, 0 }  ][line width=0.08]  [draw opacity=0] (8.93,-4.29) -- (0,0) -- (8.93,4.29) -- cycle    ;
\draw   (76,189) .. controls (76,186.79) and (77.79,185) .. (80,185) .. controls (82.21,185) and (84,186.79) .. (84,189) .. controls (84,191.21) and (82.21,193) .. (80,193) .. controls (77.79,193) and (76,191.21) .. (76,189) -- cycle ;
\draw   (76,300.18) .. controls (76,297.97) and (77.79,296.18) .. (80,296.18) .. controls (82.21,296.18) and (84,297.97) .. (84,300.18) .. controls (84,302.39) and (82.21,304.18) .. (80,304.18) .. controls (77.79,304.18) and (76,302.39) .. (76,300.18) -- cycle ;
\draw    (80,193) -- (80,220.09) ;
\draw    (80,269.09) -- (80,296.18) ;
\draw    (113.33,189) -- (84,189) ;
\draw   (113,169.5) -- (183,169.5) -- (183,320.5) -- (113,320.5) -- cycle ;
\draw    (113.33,300.18) -- (84,300.18) ;
\draw   (206,180) .. controls (206,177.79) and (207.79,176) .. (210,176) .. controls (212.21,176) and (214,177.79) .. (214,180) .. controls (214,182.21) and (212.21,184) .. (210,184) .. controls (207.79,184) and (206,182.21) .. (206,180) -- cycle ;
\draw   (206,310) .. controls (206,307.79) and (207.79,306) .. (210,306) .. controls (212.21,306) and (214,307.79) .. (214,310) .. controls (214,312.21) and (212.21,314) .. (210,314) .. controls (207.79,314) and (206,312.21) .. (206,310) -- cycle ;
\draw    (206,180) -- (182.67,180) ;
\draw    (206,310) -- (182.67,310) ;
\draw    (309.75,180) -- (297.75,180) ;
\draw    (388.67,180) -- (359.21,180) ;
\draw    (369.5,179.75) -- (369.5,217.75) ;
\draw [shift={(369.5,220.75)}, rotate = 270] [fill={rgb, 255:red, 0; green, 0; blue, 0 }  ][line width=0.08]  [draw opacity=0] (8.93,-4.29) -- (0,0) -- (8.93,4.29) -- cycle    ;
\draw   (369.8,265.79) .. controls (365.22,265.79) and (361.5,261.98) .. (361.5,257.27) .. controls (361.5,252.57) and (365.22,248.76) .. (369.8,248.76) .. controls (374.38,248.76) and (378.1,252.57) .. (378.1,257.27) .. controls (378.1,261.98) and (374.38,265.79) .. (369.8,265.79) -- cycle (369.8,274.3) -- (369.8,265.79) (369.8,240.25) -- (369.8,248.76) ;
\draw    (369.5,220.75) -- (369.8,240.25) ;
\draw    (369.8,274.3) -- (370,309.8) ;
\draw    (612,310) -- (214,310) ;
\draw    (370.1,251.79) -- (369.86,262.79) ;
\draw [shift={(369.8,265.79)}, rotate = 271.23] [fill={rgb, 255:red, 0; green, 0; blue, 0 }  ][line width=0.08]  [draw opacity=0] (8.93,-4.29) -- (0,0) -- (8.93,4.29) -- cycle    ;
\draw   (388.67,180) .. controls (388.67,177.79) and (390.46,176) .. (392.67,176) .. controls (394.88,176) and (396.67,177.79) .. (396.67,180) .. controls (396.67,182.21) and (394.88,184) .. (392.67,184) .. controls (390.46,184) and (388.67,182.21) .. (388.67,180) -- cycle ;
\draw   (412.67,218.98) -- (412.82,241.49) (421.05,246.44) -- (404.65,246.54) (421.02,241.43) -- (404.62,241.54) (412.85,246.49) -- (413,269) ;
\draw    (430.12,180) -- (396.67,180) ;
\draw    (412.56,179.98) -- (412.67,218.98) ;
\draw    (413,269) -- (413,310) ;
\draw    (450.33,180) -- (433.12,180) ;
\draw [shift={(430.12,180)}, rotate = 360] [fill={rgb, 255:red, 0; green, 0; blue, 0 }  ][line width=0.08]  [draw opacity=0] (8.93,-4.29) -- (0,0) -- (8.93,4.29) -- cycle    ;
\draw   (450.33,180) -- (460.91,180) -- (463.26,174.5) -- (467.96,185.5) -- (472.66,174.5) -- (477.36,185.5) -- (482.06,174.5) -- (486.76,185.5) -- (491.46,174.5) -- (496.16,185.5) -- (498.51,180) -- (509.08,180) ;
\draw   (509.08,180) -- (517.99,180) .. controls (518.11,177.15) and (519.35,174.71) .. (521.1,173.86) .. controls (522.85,173.01) and (524.75,173.91) .. (525.9,176.14) .. controls (526.78,177.88) and (527.14,180.13) .. (526.89,182.31) .. controls (526.89,183.17) and (526.45,183.86) .. (525.9,183.86) .. controls (525.35,183.86) and (524.91,183.17) .. (524.91,182.31) .. controls (524.66,180.13) and (525.02,177.88) .. (525.9,176.14) .. controls (526.93,174.29) and (528.36,173.24) .. (529.86,173.24) .. controls (531.35,173.24) and (532.79,174.29) .. (533.81,176.14) .. controls (534.7,177.88) and (535.06,180.13) .. (534.8,182.31) .. controls (534.8,183.17) and (534.36,183.86) .. (533.81,183.86) .. controls (533.27,183.86) and (532.82,183.17) .. (532.82,182.31) .. controls (532.57,180.13) and (532.93,177.88) .. (533.81,176.14) .. controls (534.84,174.29) and (536.27,173.24) .. (537.77,173.24) .. controls (539.27,173.24) and (540.7,174.29) .. (541.73,176.14) .. controls (542.61,177.88) and (542.97,180.13) .. (542.72,182.31) .. controls (542.72,183.17) and (542.27,183.86) .. (541.73,183.86) .. controls (541.18,183.86) and (540.74,183.17) .. (540.74,182.31) .. controls (540.48,180.13) and (540.84,177.88) .. (541.73,176.14) .. controls (542.87,173.91) and (544.78,173.01) .. (546.53,173.86) .. controls (548.28,174.71) and (549.51,177.15) .. (549.64,180) -- (558.54,180) ;
\draw    (558.54,180) -- (609,180) ;
\draw [shift={(612,180)}, rotate = 180] [fill={rgb, 255:red, 0; green, 0; blue, 0 }  ][line width=0.08]  [draw opacity=0] (8.93,-4.29) -- (0,0) -- (8.93,4.29) -- cycle    ;
\draw    (629,180) -- (611,180) ;
\draw   (629,180) .. controls (629,177.79) and (630.79,176) .. (633,176) .. controls (635.21,176) and (637,177.79) .. (637,180) .. controls (637,182.21) and (635.21,184) .. (633,184) .. controls (630.79,184) and (629,182.21) .. (629,180) -- cycle ;
\draw  [dash pattern={on 4.5pt off 4.5pt}] (60.86,143) -- (446.86,143) -- (446.86,333) -- (60.86,333) -- cycle ;
\draw  [dash pattern={on 4.5pt off 4.5pt}] (453.86,143) -- (562.86,143) -- (562.86,333) -- (453.86,333) -- cycle ;
\draw  [dash pattern={on 4.5pt off 4.5pt}] (568.97,333) -- (568.97,143) -- (638.97,143) ;
\draw  [dash pattern={on 4.5pt off 4.5pt}]  (568.97,333) -- (638.97,333) ;
\draw    (200.86,277) -- (200.86,215) ;
\draw [shift={(200.86,213)}, rotate = 90] [color={rgb, 255:red, 0; green, 0; blue, 0 }  ][line width=0.75]    (10.93,-3.29) .. controls (6.95,-1.4) and (3.31,-0.3) .. (0,0) .. controls (3.31,0.3) and (6.95,1.4) .. (10.93,3.29)   ;

\draw (124,234.71) node [anchor=north west][inner sep=0.75pt]   [align=left] {Buck $\displaystyle i$};
\draw (201,238.71) node [anchor=north west][inner sep=0.75pt]    {$V_{ti}$};
\draw (233,151.71) node [anchor=north west][inner sep=0.75pt]    {$R_{ti}$};
\draw (282,151.71) node [anchor=north west][inner sep=0.75pt]    {$I_{ti}$};
\draw (324,150.71) node [anchor=north west][inner sep=0.75pt]    {$L_{ti}$};
\draw (343,201.71) node [anchor=north west][inner sep=0.75pt]    {$I_{Li}$};
\draw (370,153.71) node [anchor=north west][inner sep=0.75pt]    {$PCC_{i}$};
\draw (384,184.71) node [anchor=north west][inner sep=0.75pt]    {$V_{i}$};
\draw (427,150.71) node [anchor=north west][inner sep=0.75pt]    {$I_{ij}$};
\draw (466.13,149.5) node [anchor=north west][inner sep=0.75pt]    {$R_{ij}$};
\draw (522,147.71) node [anchor=north west][inner sep=0.75pt]    {$L_{ij}$};
\draw (596,149.71) node [anchor=north west][inner sep=0.75pt]    {$I_{ji}$};
\draw (423,237.71) node [anchor=north west][inner sep=0.75pt]    {$C_{ti}$};
\draw (624,184.71) node [anchor=north west][inner sep=0.75pt]    {$V_{j}$};
\draw (232,115.57) node [anchor=north west][inner sep=0.75pt]   [align=left] {{\large \textbf{DGU }$\displaystyle i$}};
\draw (451,113.57) node [anchor=north west][inner sep=0.75pt]   [align=left] {{\large \textbf{Line }$\displaystyle ij$ \textbf{and }$\displaystyle ji$}};

\end{tikzpicture}